\documentclass[a4paper,11pt]{article}
\usepackage{fontspec}
\usepackage{amsmath,amssymb,amsthm,booktabs,array}
\usepackage[margin=24mm]{geometry}
\usepackage{xurl}
\usepackage{tikz}
\usepackage[hidelinks,unicode]{hyperref}
\hypersetup{pdftitle={Threshold Saturation from a Bounded Region in Multidimensional Spatially Coupled Codes over the BEC},pdfauthor={Kenta Kasai},pdfsubject={Decoding from a bounded shortened region over the BEC}}
\newtheorem{theorem}{Theorem}[section]
\newtheorem{lemma}[theorem]{Lemma}
\newtheorem{proposition}[theorem]{Proposition}
\theoremstyle{definition}
\newtheorem{definition}[theorem]{Definition}
\theoremstyle{remark}
\newtheorem{remark}[theorem]{Remark}
\newcommand{\Z}{\mathbb Z}
\newcommand{\R}{\mathbb R}
\newcommand{\F}{\mathcal F}
\newcommand{\eps}{\epsilon}
\newcommand{\pot}{\epsilon_{\mathrm{pot}}}

\newcommand{\norm}[1]{\left\lVert#1\right\rVert}

\title{Threshold Saturation from a Bounded Region\\in Multidimensional Spatially Coupled Codes over the BEC}
\author{Kenta Kasai\\Institute of Science Tokyo}
\date{}
\begin{document}
\maketitle
\begin{abstract}
We prove that a bounded shortened region initiates decoding throughout multidimensional spatially coupled regular LDPC and MacKay--Neal (MN) codes over the binary erasure channel (BEC). In every fixed finite dimension, uniform hypercube coupling permits the coupling width and shortened region to be chosen independently of the total number $V$ of spatial positions. At fixed widths and dimension $d>1$, replacing a shortened slab by a bounded hypercube reduces the shortening fraction from order $V^{-1/d}$ to order $V^{-1}$, with the same improvement in the shortening term of the check-count rate bound. Regular LDPC codes decode below their uncoupled potential threshold. MN codes achieve capacity for every integer degree choice $\ell>r\geq2$, $g\geq2$: their actual transmitted rates tend to $r/\ell$ and their average bit-erasure probabilities under sum-product decoding vanish below $1-r/\ell$. The proof combines an endpoint-potential identity, removal of an auxiliary constraint, finite-time estimates uniform in direction, and a curvature comparison that transfers flat-boundary progress to expanding balls. For MN codes, elementary inequalities establish fixed-point positivity for all these degrees. Two-dimensional density-evolution examples illustrate the dependence on the initial shortened region; the general sufficient constants are not evaluated numerically.
\end{abstract}

\section{Introduction}\label{sec:intro}
Spatially coupled LDPC constructions originate in convolutional codes with low-density parity-check matrices~\cite{Conv1999}. Decoding can progress from a favorable boundary and exceed the iterative decoding threshold of the uncoupled ensemble. Threshold saturation on the binary erasure channel (BEC) relates this improvement to the potential of the uncoupled recursion~\cite{KRU2011,Maxwell2014}. Universal capacity results on binary-input memoryless symmetric channels and density-valued potential proofs extend the one-dimensional theory~\cite{KRU2013,Kumar2014}. Protograph constructions connect this theory to code design and distance properties~\cite{Mitchell2015}. Multidimensional coupling and shortening of a local hypercube were introduced in~\cite{MD2013}. A region whose size is independent of the system size would make the fraction of shortened positions vanish as the system grows.

The proof follows five steps: positive potentials at nonzero homogeneous fixed points, advancement of a flat boundary, expansion of a ball, decoding throughout the torus, and, for MN codes, capacity achievement. Figure~\ref{fig:proofstory} illustrates these implications. An exact potential identity supplies the passage from fixed points to spatial progress; finite dependence then makes this progress possible from a finite shortened region. The argument applies to both scalar and vector density evolution.

For MacKay--Neal (MN) codes, spatial coupling was studied in~\cite{MN2011}, with BEC capacity results for particular degrees in~\cite{MN2013,MN2014}. Exact certificates also establish the classical MN interior fixed-point inequality for $(\ell,r,g)=(k,j,k)$, $2\leq j<k\leq60$, in a recent quantum-code application~\cite[Theorem 4]{Kasai2026}. We prove the fixed-point inequality required here for every integer $\ell>r\geq2$, $g\geq2$. The proof compares algebraic expressions for the potential along fixed-point branches and reduces them to $r=g=2$. It does not compare density-evolution trajectories at different degrees.

\subsection{Main results}\label{sec:main}
Fix an integer dimension $d\geq1$. Bold symbols denote vectors or vector-valued profiles; their components use ordinary type. Write $\norm{\mathbf{x}}_2$ for the Euclidean norm and $S^{d-1}$ for the unit sphere. The offset from a variable position $\mathbf{i}$ to a check position $\mathbf{i}+\mathbf{k}$ is uniform on
\begin{equation}\label{eq:offsets}
 K_w=\{-w,\ldots,w\}^d,\qquad
 p_w(\mathbf{k})=\frac{\mathbf1_{K_w}(\mathbf{k})}{(2w+1)^d}.
\end{equation}
Definition~\ref{def:ensemble} (socket allocation and shortening) specifies the graph ensemble and shortening. The conventional offsets $\{0,\ldots,2w\}^d$ give the same construction after relabeling check positions.

\begin{theorem}[Local shortening of regular LDPC codes]\label{thm:main}
Fix $d\geq1$, $d_v\geq3$, and $d_c>d_v$. For every $0\leq\eps<\pot(d_v,d_c)$, where the potential threshold in Definition~\ref{def:potential} is the largest erasure parameter with nonnegative scalar potential on $[0,1]$, there exist finite constants $w_0\in\Z_{\geq1}$ and $R_0>0$ with the following property. For any integer $w\geq w_0$, put
\begin{equation}\label{eq:seedside}
 s(w)=2\lceil wR_0\rceil+1.
\end{equation}
On the torus $(\Z/L\Z)^d$, where $L>s(w)+2w$, shorten all variables in a hypercube of side $s(w)$. Sum-product density evolution on the BEC$(\eps)$ then converges uniformly to zero on the torus, for both message and posterior erasure probabilities.
\end{theorem}

\begin{theorem}[Local shortening of MN codes below capacity]\label{thm:mnmain}
Fix integers $d\geq1$, $\ell>r\geq2$, and $g\geq2$. For every $0\leq\eps<1-r/\ell$, there exist finite $w_0$ and $R_0>0$, independent of $L$, such that the MN ensemble in Section~\ref{sec:mnmodel}, with offsets~\eqref{eq:offsets}, any $w\geq w_0$, and both variable types shortened in a hypercube of side~\eqref{eq:seedside}, has density evolution converging uniformly to zero on every torus with $L>s(w)+2w$. There are sequences of such codes whose transmitted rates tend to $r/\ell$ and whose average transmitted-bit erasure probabilities under sum-product decoding vanish at every fixed $\eps<1-r/\ell$.
\end{theorem}

The restriction $d_c>d_v$ ensures a positive regular-code design rate; the contraction argument uses $d_v\geq3$. At $\eps=0$, regular-code messages are initially zero, and MN decoding follows by comparison with any positive parameter below capacity. The constants are chosen separately for the two theorems and may depend on the fixed dimension, degrees, and channel parameter. A bounded shortened region means that, after fixing these parameters and a coupling width, its size is independent of $L$. In the MN capacity-achieving sequence, the degrees remain fixed while the coupling width and shortened size may vary as the channel approaches the capacity boundary. The sufficient choices depend on advancement times whose finiteness is proved, but whose numerical size is not evaluated (Section~\ref{sec:quantitative}). Density evolution first takes the number of variables per position to infinity at fixed finite geometry and iteration count. Sections~\ref{sec:rate} and~\ref{sec:mnrate} justify the passage to finite code sequences and their rates. The conclusions concern average bit erasures under the sum-product algorithm (SPA), also called belief propagation (BP).

This potential threshold equals the Maxwell threshold for the positive-rate regular ensembles considered here~\cite[Lemma 49]{Maxwell2014}; the theorem is stated directly in terms of the potential criterion. For example, regular $(3,6)$ codes have potential threshold approximately $0.48815$, below the capacity boundary $0.5$ at rate $1/2$. MN degrees $(\ell,r,g)=(4,2,2),(5,3,4)$, and $(6,4,2)$ are all included in Theorem~\ref{thm:mnmain} (MN decoding and capacity); their capacity boundaries are $1/2$, $2/5$, and $1/3$, respectively.

\subsection{An accessible outline of the proof}\label{sec:outline}
Figure~\ref{fig:proofstory} shows the logical order of the proof. Steps 1--4 apply to both code families. Step 5 gives capacity achievement for MN codes; regular LDPC codes are treated below their potential threshold. Table~\ref{tab:roadmap} gives the corresponding mathematical statements. We prove the common spatial implications first and then verify their code-specific hypotheses in Sections~\ref{sec:regular} and~\ref{sec:mn}.

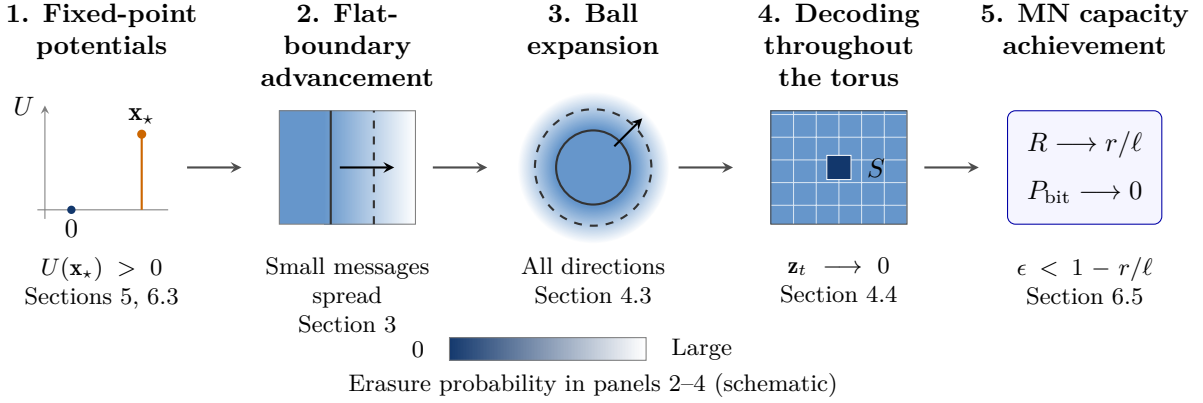
\begin{figure}[t]
\centering
\def\StoryOne{1. Fixed-point\\potentials}
\def\StoryTwo{2. Flat-boundary\\advancement}
\def\StoryThree{3. Ball\\expansion}
\def\StoryFour{4. Decoding\\throughout the torus}
\def\StoryFive{5. MN capacity\\achievement}
\def\StoryOneRef{Sections~\ref{sec:regular}, \ref{sec:mnpositivity}}
\def\StoryTwoNote{Small messages\\spread}
\def\StoryTwoRef{Section~\ref{sec:core}}
\def\StoryThreeNote{All directions}
\def\StoryThreeRef{Section~\ref{sec:ball}}
\def\StoryFourRef{Section~\ref{sec:torus}}
\def\StoryFiveRef{Section~\ref{sec:mnrate}}
\def\StoryColorSmall{0}
\def\StoryColorLarge{Large}
\def\StoryColorMeaning{Erasure probability in panels 2--4 (schematic)}
\definecolor{StoryBlue}{RGB}{75,132,194}
\definecolor{StoryZero}{RGB}{18,55,108}
\colorlet{StorySmall}{StoryBlue!85!white}
\pgfdeclareradialshading{storyradial}{\pgfpointorigin}{color(0bp)=(StorySmall);color(13bp)=(StorySmall);color(25bp)=(white)}
\begin{tikzpicture}[x=1cm,y=1cm,>=stealth,font=\small,
 title/.style={align=center,text width=2.9cm,anchor=north,font=\small\bfseries},
 detail/.style={align=center,text width=2.9cm,anchor=north,font=\footnotesize}]
\path[use as bounding box] (-1.48,-2.8) rectangle (14.48,2.58);
\foreach \xx/\tt in {0/\StoryOne,3.25/\StoryTwo,6.5/\StoryThree,9.75/\StoryFour,13/\StoryFive}
  \node[title] at (\xx,2.5) {\tt};
\begin{scope}
 \draw[->,gray] (-.75,-.48)--(-.75,1.02) node[left,black] {$U$};
 \draw[gray] (-.85,-.36)--(.88,-.36);
 \fill[StoryZero] (-.4,-.36) circle (.055) node[below,black] {$0$};
 \draw[orange!80!black,thick] (.53,-.36)--(.53,.64);
 \fill[orange!80!black] (.53,.64) circle (.06) node[above,black] {$\mathbf{x}_\star$};
 \node[detail] at (0,-.85) {$U(\mathbf{x}_\star)>0$\\\StoryOneRef};
\end{scope}
\begin{scope}[shift={(3.25,0)}]
 \fill[StorySmall] (-.9,-.55) rectangle (-.3,.95);
 \shade[left color=StorySmall,right color=white] (-.3,-.55) rectangle (.9,.95);
 \draw[gray] (-.9,-.55) rectangle (.9,.95);
 \draw[black!80,thick] (-.22,-.55)--(-.22,.95);
 \draw[black!80,thick,dashed] (.35,-.55)--(.35,.95);
 \draw[->,thick] (-.1,.2)--(.62,.2);
 \node[detail] at (0,-.85) {\StoryTwoNote\\\StoryTwoRef};
\end{scope}
\begin{scope}[shift={(6.5,.2)}]
 \shade[shading=storyradial] (0,0) circle (1.0);
 \draw[black!80,thick] (0,0) circle (.49);
 \draw[black!80,thick,dashed] (0,0) circle (.77);
 \draw[->,thick] (.28,.28)--(.65,.65);
 \node[detail] at (0,-1.05) {\StoryThreeNote\\\StoryThreeRef};
\end{scope}
\begin{scope}[shift={(9.75,0)}]
 \fill[StorySmall] (-.9,-.55) rectangle (.9,.95);
 \draw[step=.3,white!80!StoryBlue,very thin] (-.9,-.55) grid (.9,.95);
 \draw[black!70] (-.9,-.55) rectangle (.9,.95);
 \filldraw[fill=StoryZero,draw=white,line width=.25pt] (-.16,.04) rectangle (.16,.36);
 \node[anchor=west] at (.24,.2) {$S$};
 \node[detail] at (0,-.85) {$\mathbf z_t\longrightarrow0$\\\StoryFourRef};
\end{scope}
\begin{scope}[shift={(13,0)}]
 \draw[rounded corners=3pt,blue!65!black,fill=blue!4] (-1.02,-.55) rectangle (1.02,.95);
 \node[align=center] at (0,.2) {$R\longrightarrow r/\ell$\\[6pt]$P_{\mathrm{bit}}\longrightarrow0$};
 \node[detail] at (0,-.85) {$\eps<1-r/\ell$\\\StoryFiveRef};
\end{scope}
\foreach \xx in {1.13,4.38,7.63,10.88}
 \draw[->,black!70,thick] (\xx,.2)--(\xx+.72,.2);
\node[anchor=east,font=\footnotesize] at (4.4,-2.2) {\StoryColorSmall};
\shade[left color=StoryZero,middle color=StorySmall,right color=white] (4.6,-2.35) rectangle (7.2,-2.05);
\draw[gray,thin] (4.6,-2.35) rectangle (7.2,-2.05);
\node[anchor=west,font=\footnotesize] at (7.4,-2.2) {\StoryColorLarge};
\node[font=\footnotesize] at (6.5,-2.68) {\StoryColorMeaning};
\end{tikzpicture}
\caption{The five-step proof outline. Panel 1 compares the zero state with a generic nonzero homogeneous fixed point. In panels 2--4, darker blue denotes smaller erasure probabilities and white denotes larger ones. The schematic shading in panels 2--3 shows the state after advancement; solid and dashed lines mark regions satisfying the same small bound before and after advancement. The dark navy endpoint of the scale denotes zero erasure probability and is also used for the shortened set $S$, known from the start. Lighter blue denotes small positive erasure probabilities. The pale blue in panel 4 shows the stage when erasure probabilities are small everywhere; the formula below states their subsequent convergence to zero. Panel 5 gives MN capacity achievement, with transmitted rate $R=k_T/n_T$ from~\eqref{eq:mnactualrate} and SPA bit-erasure probability $P_{\mathrm{bit}}=E_{T,t}(\Gamma,\eps)$ from~\eqref{eq:mnratedelivery}. The spatial panels are two-dimensional illustrations of an argument valid in every fixed finite dimension. Small-message contraction is also required.}\label{fig:proofstory}
\end{figure}

\noindent\textbf{1. Positive fixed-point potentials.}
A homogeneous fixed point has the same message vector at every position and is unchanged by density evolution. The complete-decoding point has potential zero. We require every other homogeneous fixed point to have positive potential, together with contraction of sufficiently small messages (potential and contraction assumptions; Definition~\ref{def:admissible}). For regular LDPC codes, positivity follows from the potential threshold in Definition~\ref{def:potential} and its strict-positivity consequence~\eqref{eq:regularpositive}. For MN codes, we separate the complete-decoding point, the other trivial fixed point, and the nontrivial fixed points. The other trivial fixed point $(1,\eps)$ has the capacity-gap potential $1-r/\ell-\eps$ in~\eqref{eq:mnbadpotential}; the degree comparisons in Section~\ref{sec:mnpositivity} establish positivity on the nontrivial fixed-point set $\mathcal F_{\mathrm{nt}}^{\mathrm{MN}}(\eps)$ defined in~\eqref{eq:mnfixedset}.

\noindent\textbf{2. Advancement of a flat boundary.}
Use the step profile~\eqref{eq:halfprofile}: the left half-line starts at the contracting bound $\boldsymbol{\alpha}$ from~\eqref{eq:smallbox}, and the right half-line starts at the upper message vector $\mathbf{b}_{\max}$ from Definition~\ref{def:admissible} (potential and contraction assumptions). Temporarily cap the left half-line at $\boldsymbol{\alpha}$ after each update. This auxiliary sequence decreases. If its limit had a nonzero right endpoint, that endpoint would be a homogeneous fixed point. The endpoint-potential identity~\eqref{eq:translationidentity} of Lemma~\ref{lem:identity} would force its potential to be nonpositive, contradicting Step 1. The auxiliary limit is therefore zero. Once a finite neighborhood is small, the cap becomes inactive. A finite-time comparison with a translated auxiliary profile then proves convergence for the original free trajectory (free half-line convergence, Proposition~\ref{prop:free}; translated comparison~\eqref{eq:freecomparison}). Thus the small-message region advances across a flat boundary under the actual density evolution.

\noindent\textbf{3. Expansion of a ball.}
The hypercube coupling region depends on direction. Sections~\ref{sec:directions} and~\ref{sec:halfspace} use finite-time continuity and a finite cover of the unit sphere to obtain a common advancement time for all directions and sufficiently fine grids. A finite number of updates depends only on a bounded neighborhood. Within this neighborhood, a sufficiently large sphere is close to a plane. The curvature estimate~\eqref{eq:curvature} therefore transfers the flat-boundary comparison to an expanding ball (ball expansion; Lemma~\ref{lem:ballexpansion}). The ball carries a small positive message bound, which is enough for the comparison; its whole interior need not be shortened anew. Remark~\ref{rem:nucleation} relates the sufficient initial radius to the critical-nucleus picture and explains why its dependence on dimension is compatible with the theorem.

\noindent\textbf{4. Decoding throughout the torus.}
Shorten a finite hypercube containing the initial ball. Its actual zero messages are bounded above by the small positive comparison profile. Translations of the ball comparison, equations~\eqref{eq:geomtranslate}--\eqref{eq:geomtranslationinduction}, make the messages small at every position on the torus. The contraction~\eqref{eq:smallbox} then sends all messages to zero. The required coupling width and shortened region are independent of the torus side length (transfer to local shortening; Theorem~\ref{thm:transfer}).

\noindent\textbf{5. Capacity achievement for MN codes.}
As the torus grows, the fraction of shortened positions vanishes. Decoding alone does not determine the transmitted rate after puncturing: Section~\ref{sec:mnrate} also controls the rank of the punctured-variable matrix. Together with the bit-erasure conclusion of Step 4, the rate bounds~\eqref{eq:mnactualrate} and~\eqref{eq:mncapacitysandwich} yield codes whose transmitted rates tend to $r/\ell$ and whose SPA bit-erasure probabilities vanish at every fixed $\eps<1-r/\ell$. This establishes Theorem~\ref{thm:mnmain} (MN decoding and capacity).

\begin{table}[tb]
\centering\small
\caption{The five steps in Figure~\ref{fig:proofstory} and their mathematical justification. Steps 1--4 apply to both code families; Step 5 gives the MN capacity statement. Every spatial comparison retains the two averages in density evolution.}\label{tab:roadmap}
\begin{tabular}{@{}>{\raggedright\arraybackslash}p{.20\linewidth}>{\raggedright\arraybackslash}p{.46\linewidth}>{\raggedright\arraybackslash}p{.28\linewidth}@{}}
\toprule
Object & Claim and role & Reference\\\midrule
1. Fixed-point potentials & Positivity at every nonzero homogeneous fixed point excludes a nonzero endpoint of the auxiliary profile; small messages contract & Definition~\ref{def:admissible} (potential and contraction assumptions); Sections~\ref{sec:regular}, \ref{sec:mnpositivity}; contraction~\eqref{eq:smallbox}\\
2. Flat boundary & The endpoint identity and removal of the temporary cap give free half-line convergence & Lemma~\ref{lem:identity} (endpoint-potential identity); Proposition~\ref{prop:free} (free half-line convergence); free-profile comparison~\eqref{eq:freecomparison}\\
3. Expanding ball & A common time for all directions and a curvature bound transfer the plane comparison to a ball & Sections~\ref{sec:directions}--\ref{sec:ball}; curvature bound~\eqref{eq:curvature}\\
4. Entire torus & Translated ball comparisons make all messages small; contraction gives convergence to zero & Section~\ref{sec:torus}; translation induction~\eqref{eq:geomtranslationinduction}; contraction~\eqref{eq:smallbox}\\
5. MN capacity & Vanishing shortening fraction, projection rank, and bit recovery give transmitted rate $r/\ell$ and vanishing bit erasure & Section~\ref{sec:mnrate}; rate lower bound~\eqref{eq:mnactualrate}; rate limit~\eqref{eq:mncapacitysandwich}\\
\bottomrule
\end{tabular}
\end{table}

\subsection{Relation to previous work}\label{sec:related}
The main contributions are the following three steps.
\begin{enumerate}
\item Lemma~\ref{lem:identity} (endpoint-potential identity) and Proposition~\ref{prop:free} (free half-line convergence) turn fixed-point positivity into convergence of a freely evolving half-line. The endpoint identity applies to bounded monotone vector profiles, and the finite-time comparison removes the temporary constraint on the favorable half-line.
\item Lemmas~\ref{lem:finiteperturbation}--\ref{lem:ballexpansion} (finite-time perturbation through ball expansion) make that progress uniform in direction, transfer it to a discrete grid, and control curvature. This proves decoding from a bounded interior region in every fixed finite dimension.
\item Lemma~\ref{lem:mnpositive} (nontrivial fixed-point positivity) proves MN fixed-point positivity for all $\ell>r\geq2$, $g\geq2$ by elementary degree comparisons. The finite-code and projection-rank arguments in Section~\ref{sec:mnrate} then establish capacity at the actual transmitted rate.
\end{enumerate}

Multidimensional coupling and local hypercube shortening were already introduced in~\cite{MD2013}. Its Proposition~1 reduces strip-shortened density evolution to one dimension, while Section~V derives the local-shortening rate bound and gives numerical two-dimensional thresholds. Our contribution is an analytic guarantee of global decoding from a bounded hypercube for every fixed dimension. In the same $d$-dimensional geometry, a slab of fixed thickness shortens order $V^{1-1/d}$ positions, whereas the hypercube shortens a constant number. Section~\ref{sec:rate} gives the exact fractions and the corresponding check-count rate penalties. A one-dimensional chain with $V$ positions also has an $O(V^{-1})$ termination fraction~\cite{MD2013}. The comparison fixes the channel, degrees, dimension, and widths before increasing $V$; the sufficient widths and sizes may be much larger for the local construction.

The continuum analysis in~\cite[Theorem 2]{Potential2013} fixes the outer boundary to a favorable potential-minimizing state. An interior region on a periodic domain instead has a curved boundary, and the rest of the profile must evolve freely. The first two contributions above address this boundary condition while controlling the discrete recursion. The potential primitives themselves are established tools for scalar and vector recursions~\cite{Maxwell2014,Vector2012}. Scalar one-dimensional propagation, spatial fixed-point integration, and comparisons between continuous and discrete recursions were developed in~\cite{Waves2015}. Here the endpoint argument is formulated for vector messages and is combined with direction-uniform estimates and a curvature bound. Section~\ref{sec:identity} explains the endpoint formulation. Related propagation from favorable regions appears in coupled mean-field models and constraint satisfaction problems~\cite{Hassani2012,Hassani2013}, with the dependence on the initial region studied in~\cite{Caltagirone2014}.

The MN comparison is independent of the finite-degree certificates in~\cite[Theorem 4]{Kasai2026}: the proof of Lemma~\ref{lem:mnpositive} (nontrivial fixed-point positivity) uses no certificate or result from that preprint. Multidimensional circulant constructions address short-cycle design~\cite{Cycles2020,Prob2024}; finite-length scaling~\cite{Olmos2015} and windowed decoding~\cite{Iyengar2013} address different performance and complexity questions. The present geometric theorem supplies a decoding guarantee for local shortening, while Section~\ref{sec:numerics} gives finite-grid density-evolution illustrations.

\section{Preliminaries and a common sufficient condition}\label{sec:prelim}
The first step in Figure~\ref{fig:proofstory} supplies the hypotheses for the spatial steps: positivity at nonzero homogeneous fixed points and contraction near zero. We formulate these conditions for scalar and vector messages together, retaining the check update between the two spatial averages. Theorem~\ref{thm:transfer} (transfer to local shortening) states their common consequence, decoding throughout a torus from a finite shortened region.

\begin{table}[tb]
\centering\small
\caption{Principal notation. The spatial dimension $d$ and message dimension $m$ are distinct.}\label{tab:notation}
\begin{tabular}{@{}p{.29\linewidth}>{\raggedright\arraybackslash}p{.64\linewidth}@{}}
\toprule
Symbol & Meaning and definition\\\midrule
$\mathbf b_{\max},\boldsymbol\alpha$ & Upper message vector and upper corner of the contracting rectangle (potential and contraction assumptions; Definition~\ref{def:admissible})\\
$\mathbf u_{-},\mathbf u_{+}$ & Left and right limits of a monotone profile (endpoint-potential identity; Lemma~\ref{lem:identity})\\
$B_R,\mathcal B,S$ & Euclidean ball (Section~\ref{sec:ball}), lattice ball and shortened set (Section~\ref{sec:torus})\\
$N_{\mathbf e},\mathcal O_{\mathbf e}$ & Advancement time and neighborhood of a direction (\eqref{eq:geomdirectionaltime}, \eqref{eq:geomneighborhood})\\
$p,q$ & Regular-code excess degrees $d_v-1,d_c-1$ (\eqref{eq:scalarfunctions})\\
$a,b;\ \chi_g,\sigma_n$ & MN erasure coordinates; auxiliary branch factor and geometric sum (Section~\ref{sec:mnpositivity})\\
$V,V_S$ & Total and shortened position counts (\eqref{eq:mncounts})\\
$\theta,\alpha_{\min},K$ & Small-message contraction factor and perturbation constants (\eqref{eq:smallbox}, \eqref{eq:geomstep})\\
$\rho,M_0,\Delta_w,\kappa$ & Support radius, projected-density bound, grid error, and distribution-coupling error (\eqref{eq:geomprojections}--\eqref{eq:geomcoupling})\\
$m_{\mathrm{dir}},N_*,t_*,J$ & Number of covering directions, largest selected time, common time, and dependence range (Section~\ref{sec:halfspace}, \eqref{eq:geomcommontime}, \eqref{eq:geomradius})\\
$w_0,R_0,s(w)$ & Sufficient width, scaled initial radius, and shortened side length (\eqref{eq:geomwidth}, \eqref{eq:geomradius}, \eqref{eq:seedside})\\
$\pot$ & Uncoupled regular-code potential threshold (Definition~\ref{def:potential})\\
\bottomrule
\end{tabular}
\end{table}

\begin{definition}[Balanced socket allocation and shortening]\label{def:ensemble}
A socket is an edge endpoint attached to a variable or check node before edges are matched. At each position, partition the sockets of each variable type uniformly among equally sized groups indexed by the offset set $K_w$ in~\eqref{eq:offsets}. Independently partition the corresponding sockets of each check type into incoming offset groups. For each $\mathbf{k}$, match the appropriate groups between positions $\mathbf{i}$ and $\mathbf{i}+\mathbf{k}$ uniformly at random. All socket counts are integral, by choosing the number of variables per position appropriately. Shortening fixes the variables in a designated set to zero and removes them from transmission. In an MN ensemble, puncturing leaves designated variables unobserved by the decoder; it does not fix them.
\end{definition}

The offset allocation is over all sockets of a type at a position; individual variables do not have prescribed offset lists. Let $M$ denote the section-size parameter specified for each ensemble below. At fixed geometry and a finite iteration count, Lemma~\ref{lem:finiteDE} (finite-iteration approximation) proves an $O(M^{-1})$ approximation to density evolution for this allocation, including its sampling without replacement. Sections~\ref{sec:regular} and~\ref{sec:mnmodel} specify the degrees and numbers of variables and checks.

\begin{definition}[Coupled recursion and potential assumptions]\label{def:admissible}
Let $\mathbf{b}_{\max}\in(0,1]^m$. Let $Q:[0,\mathbf{b}_{\max}]\to[0,1]^m$ and $f:[0,1]^m\to[0,\mathbf{b}_{\max}]$ be continuously differentiable, componentwise order-preserving maps, with $Q(0)=f(0)=0$. Write $T=f\circ Q$. Suppose a positive diagonal matrix $D$ and differentiable functions $G,F$ satisfy
\begin{equation}\label{eq:primitives}
 \nabla G(\mathbf{x})=DQ(\mathbf{x}),\qquad \nabla F(\mathbf{y})=Df(\mathbf{y}),\qquad G(0)=F(0)=0.
\end{equation}
The homogeneous potential is
\begin{equation}\label{eq:vectorpotential}
 U(\mathbf{x})=Q(\mathbf{x})^{\mathsf T}D\mathbf{x}-G(\mathbf{x})-F(Q(\mathbf{x})).
\end{equation}
Assume that some $0<\boldsymbol{\alpha}\leq \mathbf{b}_{\max}$ and $0<\theta<1$ satisfy
\begin{equation}\label{eq:smallbox}
 T(t\boldsymbol{\alpha})\leq\theta t\boldsymbol{\alpha}\quad(0\leq t\leq1),
\end{equation}
and that
\begin{equation}\label{eq:positivefixed}
 U(\mathbf{x}_\star)>0\quad\text{whenever }0\ne \mathbf{x}_\star=T(\mathbf{x}_\star),\quad \mathbf{x}_\star\in[0,\mathbf{b}_{\max}].
\end{equation}
All vector inequalities, minima, and intervals are componentwise; $0$ denotes the zero vector when appropriate. Maps and operators retain ordinary type. The norm $\|\cdot\|_\infty$ takes the supremum over message components and, for profiles, spatial positions. This differs from the Euclidean norm $\|\cdot\|_2$ used for spatial distances.
\end{definition}

The positivity and contraction assumptions have different roles. The upper vector $\mathbf b_{\max}$ in Definition~\ref{def:admissible} (potential and contraction assumptions) bounds all allowed erasure messages, whereas $\boldsymbol\alpha$ specifies the smaller rectangle from which contraction is guaranteed. The contraction factor $\theta$ in~\eqref{eq:smallbox} controls repeated updates once every position is in that rectangle. In contrast, fixed-point positivity~\eqref{eq:positivefixed} is used to bring new positions into it. The homogeneous map $T=f\circ Q$ updates a spatially constant message; the spatial update $\F_\lambda$ in Definition~\ref{def:DE} (shortened density evolution) also averages neighboring messages. These two updates agree on constant profiles.

The primitives follow the potential formulation for coupled recursions~\cite{Maxwell2014,Vector2012}. We use these gradient relations at a fixed channel parameter; the sufficient hypotheses for our transfer theorem are those stated in Definition~\ref{def:admissible}. Condition~\eqref{eq:positivefixed} concerns fixed points only. The geometric proof does not require a sign for $U$ at other vectors. Since $f$ maps into $[0,\mathbf{b}_{\max}]$, the upper vector satisfies $T(\mathbf{b}_{\max})\leq \mathbf{b}_{\max}$; equality is unnecessary.

\begin{definition}[Averaging and shortened density evolution]\label{def:DE}
For a symmetric probability measure $\lambda$ on $\R$ or $\R^d$, set
\begin{equation}\label{eq:operator}
 A_\lambda \mathbf{u}(\mathbf{x})=\int \mathbf{u}(\mathbf{x}+\mathbf{z})\,d\lambda(\mathbf{z}),\qquad
 \F_\lambda \mathbf{u}=f\bigl(A_\lambda Q(A_\lambda \mathbf{u})\bigr).
\end{equation}
The discrete distribution on scale $1/w$ is
\begin{equation}\label{eq:discretemeasure}
 \nu_w=\frac1{(2w+1)^d}\sum_{\mathbf{k}\in K_w}\delta_{\mathbf{k}/w}.
\end{equation}
For integer spatial coordinates $\mathbf i\in\Z^d$, define
\begin{equation}\label{eq:latticeoperator}
 A_w\mathbf u(\mathbf i)=\frac{1}{(2w+1)^d}\sum_{\mathbf k\in K_w}\mathbf u(\mathbf i+\mathbf k),
 \qquad \F_w\mathbf u=f\bigl(A_wQ(A_w\mathbf u)\bigr).
\end{equation}
The scaled profile $\widehat{\mathbf u}(\mathbf i/w)=\mathbf u(\mathbf i)$ satisfies
\begin{equation}\label{eq:scalerelation}
 (\F_{\nu_w}\widehat{\mathbf u})(\mathbf i/w)=(\F_w\mathbf u)(\mathbf i).
\end{equation}
On the integer torus $(\Z/L\Z)^d$, $\F_S$ applies $\F_w$ with all indices reduced modulo $L$ and then sets all components on the shortened set $S$ to zero. Equivalently, the scaled torus has side length $L/w$. The initial profile is zero on $S$ and $\mathbf{b}_{\max}$ elsewhere.
\end{definition}

The inner average collects variable messages at a check, and the outer average collects the resulting check messages at a variable. One offset is generally reflected; symmetry of the offset distribution~\eqref{eq:offsets} permits both averages to use the same distribution, as in multidimensional BEC density evolution~\cite[Section III]{MD2013}. Both $\F_\lambda$ and $\F_S$ preserve order. If $\mathbf{u}\leq t\boldsymbol{\alpha}$, averaging and the contraction condition~\eqref{eq:smallbox} give $\F_\lambda \mathbf{u}\leq\theta t\boldsymbol{\alpha}$. A distribution supported in a ball of radius $\rho$ gives a dependence radius of $2\rho$ per update.

\begin{theorem}[Transfer from positive fixed-point potentials]\label{thm:transfer}
Under Definition~\ref{def:admissible} (potential and contraction assumptions), for each fixed $d\geq1$ there exist finite $w_0$ and $R_0$ such that the shortened recursion in Definition~\ref{def:DE} (shortened density evolution), with $w\geq w_0$, $s(w)$ from~\eqref{eq:seedside}, and $L>s(w)+2w$, converges uniformly to zero on the torus.
\end{theorem}
Sections~\ref{sec:core} and~\ref{sec:geometry} prove the theorem. The constants can be expressed through finitely many continuum half-line advancement times; Section~\ref{sec:quantitative} gives sufficient choices.

\section{A potential identity and free half-line convergence}\label{sec:core}
The passage from fixed-point potentials to flat-boundary advancement is Step 2 in Figure~\ref{fig:proofstory}. An exact identity excludes a nonzero endpoint of an auxiliary capped profile. Removing the cap by a finite-time comparison proves free half-line convergence, which supplies the planar comparison used for ball expansion in Section~\ref{sec:ball}.

\subsection{The endpoint identity}\label{sec:identity}
The potential and translation method is inherited from~\cite{Maxwell2014,Vector2012}; in particular, \cite[Lemma 24]{Maxwell2014} bounds a finite discrete shift of a coupled scalar potential. Spatial fixed-point integration also relates potential values to averaged scalar profiles in~\cite{Waves2015}. The form below is an exact whole-line endpoint identity for bounded monotone vector profiles. It allows jumps in the profile and nonzero limits at infinity, so the profile itself need not be integrable. Only derivatives of its averaged profiles enter the integrable expression. This formulation applies to the auxiliary inequality $\mathbf u\leq\F\mathbf u$ and supplies the sign needed for cap removal, without assuming that $\mathbf u$ is a stationary traveling wave.

In this section, $\omega$ is an even bounded probability density supported in $[-\rho,\rho]$. Write $A=A_\omega$ and $\F=\F_\omega$.

\begin{lemma}[Exact endpoint-potential identity]\label{lem:identity}
Let $\mathbf{u}:\R\to[0,\mathbf{b}_{\max}]$ be componentwise nondecreasing, with endpoints $\mathbf{u}_{-}$ and $\mathbf{u}_{+}$. Set $\mathbf{z}=A\mathbf{u}$, $\mathbf{v}=Q(\mathbf{z})$, and $\mathbf{y}=A\mathbf{v}$. Then
\begin{equation}\label{eq:translationidentity}
 U(\mathbf{u}_{+})-U(\mathbf{u}_{-})=\int_{\R}\mathbf{y}'(s)^{\mathsf T}D\,[\mathbf{u}(s)-\F \mathbf{u}(s)]\,ds.
\end{equation}
In particular, $\mathbf{u}\leq\F \mathbf{u}$ implies $U(\mathbf{u}_{+})\leq U(\mathbf{u}_{-})$.
\end{lemma}

In the endpoint identity~\eqref{eq:translationidentity}, $\mathbf y=AQ(A\mathbf u)$ is the averaged check-message profile defined in Lemma~\ref{lem:identity} (endpoint-potential identity). Its derivative is componentwise nonnegative because averaging and the check map preserve spatial order. Thus the integral compares the update residual $\mathbf u-\F\mathbf u$ with a nonnegative weight. A nonpositive residual forces a nonpositive potential difference between the right and left endpoints. This is the precise connection between the spatial profile and the homogeneous potential~\eqref{eq:vectorpotential}.
\begin{proof}
Each component of the distributional derivative $d\mathbf{u}$ is a finite nonnegative measure. Convolution with the bounded density makes $\mathbf{z}$ Lipschitz, with $\mathbf{z}'\geq0$ and $\mathbf{z}'\in L^1(\R)$. The order and differentiability of $Q$ imply the same properties for $\mathbf{v}$; consequently $\mathbf{y}'=A\mathbf{v}'$ is nonnegative and integrable. Thus the calculation differentiates averaged profiles and allows jumps in $\mathbf{u}$.

Consider
\begin{equation}\label{eq:profileenergy}
 H(s)=\mathbf{v}(s)^{\mathsf T}D\mathbf{z}(s)-G(\mathbf{z}(s))-F(\mathbf{y}(s)).
\end{equation}
The gradient identity for $G$ in the primitive relations~\eqref{eq:primitives} cancels the terms containing $\mathbf{z}'$, while the gradient identity for $F$ gives
\begin{equation}\label{eq:energyderivative}
 H'(s)=\mathbf{v}'(s)^{\mathsf T}D\mathbf{z}(s)-\mathbf{y}'(s)^{\mathsf T}Df(\mathbf{y}(s)).
\end{equation}
Evenness of $\omega$ and Fubini's theorem yield
\begin{equation}\label{eq:selfadjoint}
 \int_{\R}\mathbf{v}'^{\mathsf T}DA\mathbf{u}\,ds
 =\int_{\R}(A\mathbf{v}')^{\mathsf T}D\mathbf{u}\,ds
 =\int_{\R}\mathbf{y}'^{\mathsf T}D\mathbf{u}\,ds.
\end{equation}
These integrals are absolutely convergent because $\mathbf{u}$ is bounded and $\mathbf{v}'\in L^1$; $\mathbf{u}$ itself need not be integrable. Averaging preserves endpoint limits, so $H(-\infty)=U(\mathbf{u}_{-})$ and $H(+\infty)=U(\mathbf{u}_{+})$. Integrating the profile-energy derivative~\eqref{eq:energyderivative} and moving the averaging operator by the symmetry identity~\eqref{eq:selfadjoint} proves the endpoint-potential identity~\eqref{eq:translationidentity}. The sign follows from $\mathbf{y}'\geq0$ and the positive diagonal entries of $D$.
\end{proof}

\subsection{Removing the auxiliary cap}\label{sec:free}
The endpoint identity rules out a nonzero limit of the constrained iteration. To obtain the flat-boundary advancement required in Figure~\ref{fig:proofstory}, the remaining task is to compare that iteration with the actual unconstrained density evolution.

\begin{proposition}[Free half-line convergence]\label{prop:free}
Define the step profile
\begin{equation}\label{eq:halfprofile}
 \mathbf{h}(s)=
 \begin{cases}\boldsymbol{\alpha},&s\leq0,\\\mathbf{b}_{\max},&s>0.\end{cases}
\end{equation}
Under Definition~\ref{def:admissible} (potential and contraction assumptions), $\F_\omega^n \mathbf{h}(s)\to0$ for every finite $s$. In particular, for any finite $s_0$, some finite $N$ satisfies $\F_\omega^N \mathbf{h}(s)\leq\boldsymbol{\alpha}/4$ for all $s\leq s_0$.
\end{proposition}

Convergence of the capped sequence alone would be insufficient: lowering messages by hand makes decoding easier. In the capped recursion~\eqref{eq:capiteration}, $\mathbf v_n$ denotes that auxiliary sequence, while $\F^n\mathbf h$ denotes the free sequence from the step profile~\eqref{eq:halfprofile}. The proof first shows that the cap eventually has no effect on $\mathbf v_n$. It then bounds the free sequence from above by a spatial translate of the auxiliary sequence in the comparison~\eqref{eq:freecomparison}. The direction of this last inequality is what transfers convergence to the actual free update.
\begin{proof}
Set
\begin{equation}\label{eq:capiteration}
 \mathbf{b}^{-}_n=T^n(\boldsymbol{\alpha}),\qquad \mathbf{b}^{+}_n=T^n(\mathbf{b}_{\max}),\qquad
 \mathbf{v}_0=\mathbf{h},\qquad \mathbf{v}_{n+1}=\min\{\mathbf{h},\F \mathbf{v}_n\}.
\end{equation}
Order preservation shows that $\mathbf{b}^{-}_n$ and $\mathbf{b}^{+}_n$ decrease, that $\mathbf{v}_n$ decreases in $n$ and is nondecreasing in $s$, and that
\begin{equation}\label{eq:capbounds}
 \mathbf{b}^{-}_n\leq \mathbf{v}_n\leq \mathbf{b}^{+}_n,\qquad \mathbf{b}^{-}_n\leq\theta^n\boldsymbol{\alpha}.
\end{equation}
The lower bound follows by induction from $\mathbf{h}\geq\boldsymbol{\alpha}\geq \mathbf{b}^{-}_{n+1}$; the upper bound follows from $\F \mathbf{v}_n\leq T(\mathbf{b}^{+}_n)=\mathbf{b}^{+}_{n+1}$.

Let $\mathbf{v}=\lim_n \mathbf{v}_n$. Dominated convergence gives $\mathbf{v}=\min\{\mathbf{h},\F \mathbf{v}\}$, hence $\mathbf{v}\leq\F \mathbf{v}$, with equality for $s>0$. Finite dependence gives $\mathbf{v}_n(s)=\mathbf{b}^{-}_n$ for $s\leq-2\rho n$. To identify the left endpoint, first fix $n$ and use $0\leq\mathbf v\leq\mathbf v_n$. Letting $s\to-\infty$ bounds that endpoint by $\mathbf b_n^{-}$; then $n\to\infty$ and the homogeneous contraction bound~\eqref{eq:capbounds} give zero. This argument does not exchange the spatial and iteration limits. Its right endpoint $\mathbf{u}_{+}$ satisfies $\mathbf{u}_{+}=T(\mathbf{u}_{+})$, because the cap is inactive on the right: let $s\to+\infty$ in $\mathbf v(s)=\F\mathbf v(s)$ and use dominated convergence in both averages and continuity of $f,Q$. If $\mathbf{u}_{+}\ne0$, Lemma~\ref{lem:identity} (endpoint-potential identity) contradicts~\eqref{eq:positivefixed}:
\begin{equation}\label{eq:capcontradiction}
 0<U(\mathbf{u}_{+})\leq U(0)=0.
\end{equation}
Thus $\mathbf{u}_{+}=0$, and monotonicity in $s$ gives $\mathbf{v}\equiv0$.

Choose $N$ with $\mathbf{v}_N(2\rho)\leq\boldsymbol{\alpha}$. For $s\leq0$, every input used in $\F \mathbf{v}_N(s)$ lies at a position at most $2\rho$, so $\F \mathbf{v}_N(s)\leq T(\boldsymbol{\alpha})\leq\boldsymbol{\alpha}$. For $s>0$, $\F \mathbf{v}_N(s)\leq \mathbf{b}_{\max}$. Hence the cap is inactive: $\F \mathbf{v}_N=\mathbf{v}_{N+1}\leq \mathbf{v}_N$. It remains inactive afterward by order preservation, and
\begin{equation}\label{eq:capinactive}
 \F^k \mathbf{v}_N=\mathbf{v}_{N+k}\longrightarrow0.
\end{equation}

It remains to compare with the free trajectory from $\mathbf{h}$. Finite dependence and~\eqref{eq:capbounds} give the exact plateaus
\begin{align}
 \F^N \mathbf{h}(s)&=\mathbf{b}^{-}_N\quad(s\leq-2\rho N),&
 \F^N \mathbf{h}&\leq \mathbf{b}^{+}_N,\label{eq:freeplateaus}\\
 \mathbf{v}_N(s)&=\mathbf{b}^{+}_N\quad(s>2\rho N),&
 \mathbf{v}_N&\geq \mathbf{b}^{-}_N.\label{eq:capplateaus}
\end{align}
For the right plateau, all inputs in the preceding updates lie on the right, and $\mathbf{b}^{+}_n\leq \mathbf{b}_{\max}$, so the cap has no effect there. If $s\leq-2\rho N$, use the left plateau in~\eqref{eq:freeplateaus}; otherwise $s+4\rho N>2\rho N$, and use the right plateau in~\eqref{eq:capplateaus}. These two cases prove
\begin{equation}\label{eq:freecomparison}
 \F^{N+k}\mathbf{h}(s)\leq \mathbf{v}_{N+k}(s+4\rho N)\longrightarrow0.
\end{equation}
Spatial monotonicity then turns convergence at $s=s_0$ into the stated bound for every $s\leq s_0$.
\end{proof}

The right plateau $\mathbf{b}^{+}_N=T^N(\mathbf{b}_{\max})$ in~\eqref{eq:capplateaus} accounts for the homogeneous evolution of the upper message bound. Keeping this value in the comparison also covers regular LDPC initial conditions for which $T(\mathbf{b}_{\max})<\mathbf{b}_{\max}$.

\section{From free half-spaces to local shortening}\label{sec:geometry}
Flat-boundary advancement now yields ball expansion and decoding throughout the torus, Steps 3--4 in Figure~\ref{fig:proofstory}. The planar comparison from Proposition~\ref{prop:free} (free half-line convergence) must first hold at a common time for every direction on a discrete grid. A curvature estimate then gives ball expansion, and translated comparisons make messages small everywhere. Contraction completes the proof of Theorem~\ref{thm:transfer} (transfer to local shortening). All vector inequalities below are componentwise.
\subsection{Directions and discrete averaging}\label{sec:directions}
Ball expansion needs planar progress in every outward normal direction. Uniform finite-time control under changes of direction and grid spacing supplies the passage from Step 2 to Step 3 in Figure~\ref{fig:proofstory}. Comparisons between continuous and discrete scalar recursions are used in~\cite{Waves2015}. The estimate below also controls a change in spatial direction and applies to vector messages.

Let $\nu$ be uniform on $[-1,1]^d$, and define the discrete distribution and its projections by

\begin{equation}\label{eq:geomprojections}
 \rho=\sqrt d,\quad
 \nu_w=\frac1{(2w+1)^d}\sum_{\mathbf{k}\in\{-w,\ldots,w\}^d}\delta_{\mathbf{k}/w},\quad
 \omega_{\mathbf{e}}=\mathcal L(\mathbf{e}\cdot \mathbf{Z}),\quad
 \omega_{w,\mathbf{e}}=\mathcal L(\mathbf{e}\cdot \mathbf{Z}_w),
\end{equation}

where $\mathbf{Z}\sim\nu$, $\mathbf{Z}_w\sim\nu_w$, and $\mathbf{e}\in S^{d-1}$. Here $\mathcal L$ denotes distribution. Every $\omega_{\mathbf{e}}$ has an even density supported in $[-\rho,\rho]$ and bounded by $M_0:=\rho/2$. Indeed, some coordinate satisfies $|e_i|\geq1/\rho$. The law of $e_i Z_i$ is uniform with density at most $\rho/2$, and convolution with the remaining independent coordinates cannot increase that bound. Thus Proposition~\ref{prop:free} (free half-line convergence) applies in every direction.

Divide each coordinate interval into $2w+1$ equal parts and map part $k$ to $k/w$. Its midpoint differs from $k/w$ by at most $1/(2w+1)$; adding its half-length gives a coordinate error at most $2/(2w+1)$. Hence we can couple the projections so that

\begin{equation}\label{eq:geomcoupling}
 |\mathbf{e}\cdot \mathbf{Z}_w-\mathbf{e}_0\cdot \mathbf{Z}|
 \leq\kappa:=\Delta_w+\rho\norm{\mathbf{e}-\mathbf{e}_0}_2,
 \qquad \Delta_w=\frac{2\rho}{2w+1}
\end{equation}

for any unit vectors $\mathbf{e},\mathbf{e}_0$. Define the step and constants

\begin{equation}\label{eq:geomstep}
 \mathbf{h}(s)=\begin{cases}\boldsymbol{\alpha},&s\leq0,\\ \mathbf{b}_{\max},&s>0,\end{cases}
 \qquad \alpha_{\min}:=\min_i\alpha_i>0,
 \qquad K:=\max\{1,L_fL_Q\},
\end{equation}

where $L_f$ bounds the Lipschitz constant of the variable update $f$ on $[0,1]^m$ and $L_Q$ bounds the Lipschitz constant of the check update $Q$ on $[0,\mathbf{b}_{\max}]$, both in the componentwise supremum norm. Their product bounds the Lipschitz constant of $\F_\lambda$ in that norm.

\begin{lemma}[Finite-time perturbation]\label{lem:finiteperturbation}
For every integer $n\geq1$, the coupling in~\eqref{eq:geomcoupling} gives

\begin{equation}\label{eq:geomfiniteerror}
 \bigl\|\F_{\omega_{w,\mathbf{e}}}^n\mathbf{h}-\F_{\omega_{\mathbf{e}_0}}^n\mathbf{h}\bigr\|_\infty
 \leq nK^n\rho\bigl(\Delta_w+\rho\norm{\mathbf{e}-\mathbf{e}_0}_2\bigr)
\end{equation}

where the supremum includes position and message component.
\end{lemma}

The error estimate~\eqref{eq:geomfiniteerror} separates two changes: $\Delta_w$ is the grid approximation error, and $\rho\norm{\mathbf e-\mathbf e_0}_2$ is the error from changing the projection direction, both defined by the coupling bound~\eqref{eq:geomcoupling}. The factor $K$ from~\eqref{eq:geomstep} bounds amplification by one update. For a fixed finite time $n$, both errors can therefore be made small. This finite-time statement, rather than an interchange of a grid limit and an infinite iteration limit, is what the directional covering argument uses.
\begin{proof}
The reference distribution function is $M_0$-Lipschitz. The perturbed distribution function lies between reference distribution functions shifted by $\pm\kappa$. Since each component of $\mathbf{h}$ has jump at most one, the two inner averages of $\mathbf{h}$ differ by at most $M_0\kappa$, and the reference average is $M_0$-Lipschitz. Changing the inner average and then the outer average therefore gives a first-update error at most $2KM_0\kappa$.

For any $H$-Lipschitz vector profile $\mathbf{u}$, let $\lambda,\lambda'$ be averaging distributions whose coupled displacements differ by at most $\kappa$, as in the projection coupling~\eqref{eq:geomcoupling}. Changing the averaging distribution gives

\begin{equation}\label{eq:geomlipschitzerror}
 \|\F_{\lambda'}\mathbf{u}-\F_\lambda \mathbf{u}\|_\infty\leq2KH\kappa
\end{equation}

by changing the two averages separately. The reference profile after $n\geq1$ updates is $K^nM_0$-Lipschitz. Writing $E_n$ for the supremum-norm error between the two profiles after $n$ updates gives

\[
 E_1\leq2KM_0\kappa,\qquad
 E_{n+1}\leq K E_n+2K^{n+1}M_0\kappa.
\]

Induction gives $E_n\leq2nK^nM_0\kappa$. Substituting the projected-density bound $M_0=\rho/2$ from Section~\ref{sec:directions} yields the finite-time error bound~\eqref{eq:geomfiniteerror}. In particular, we never assume that the initial step itself is Lipschitz.
\end{proof}

\subsection{A common time for all directions}\label{sec:halfspace}
The directional estimates must hold at a common time before they can be applied around one sphere. A finite cover supplies finitely many advancement times, and the residual-time comparison below turns them into one bound valid in every direction.

By Proposition~\ref{prop:free} (free half-line convergence) and spatial monotonicity, each direction $\mathbf{e}$ has an integer $N_{\mathbf{e}}$ such that

\begin{equation}\label{eq:geomdirectionaltime}
 N_{\mathbf{e}}\geq\max\left\{1,\left\lceil\frac{\log(1/2)}{\log\theta}\right\rceil\right\},
 \qquad \F_{\omega_{\mathbf{e}}}^{N_{\mathbf{e}}}\mathbf{h}(s)\leq\boldsymbol{\alpha}/4\quad(s\leq2).
\end{equation}

The lower bound on $N_{\mathbf{e}}$ will also control the deep interior of the ball. Around each reference direction choose the open neighborhood

\begin{equation}\label{eq:geomneighborhood}
 \mathcal O_{\mathbf{e}}=\left\{\mathbf{e}'\in S^{d-1}:\norm{\mathbf{e}'-\mathbf{e}}_2<
 \frac{\alpha_{\min}}{8N_{\mathbf{e}}K^{N_{\mathbf{e}}}\rho^2}\right\}.
\end{equation}

By compactness, select a finite subcover indexed by $1\leq j\leq m_{\mathrm{dir}}$. Its directions and times are $\mathbf e_j$ and $N_j$. Set

\begin{equation}\label{eq:geomwidth}
 w_0=\max\left\{2,\lceil\rho\rceil,
 \left\lceil\frac{8d}{\alpha_{\min}}\max_j(N_jK^{N_j})-\frac12\right\rceil\right\}.
\end{equation}

For $w\geq w_0$, the sufficient width choice~\eqref{eq:geomwidth} gives $\Delta_w\leq \alpha_{\min}/(8N_jK^{N_j}\rho)$ for every $j$. For any direction $\mathbf{e}$, select a covering $\mathbf{e}_j$. The grid and direction errors in~\eqref{eq:geomfiniteerror} are each at most $\alpha_{\min}/8$ at time $N_j$, so

\begin{equation}\label{eq:geomdiscreteadvance}
 \F_{\omega_{w,\mathbf{e}}}^{N_j}\mathbf{h}(s)\leq\boldsymbol{\alpha}/2\qquad(s\leq2).
\end{equation}

\begin{lemma}[Uniform half-space advancement]\label{lem:commonhalfspace}
Define

\begin{equation}\label{eq:geomcommontime}
 N_*:=\max_jN_j,\qquad t_*:=N_*\lceil\rho N_*+1\rceil.
\end{equation}

For every integer $t\geq t_*$, every $w\geq w_0$, and every direction $\mathbf{e}$,

\begin{equation}\label{eq:geomhalfadvance}
 \F_{\omega_{w,\mathbf{e}}}^t \mathbf{h}(s)\leq\boldsymbol{\alpha}/2\qquad(s\leq2).
\end{equation}

\end{lemma}

The common time $t_*$ in~\eqref{eq:geomcommontime} does more than take the largest directional time $N_*$. The advance bound~\eqref{eq:geomdiscreteadvance} initially holds at times $N_j$ that depend on the covering direction. Moreover, the free trajectory need not decrease at every iteration. The proof repeats the advance in blocks of $N_j$ updates and reserves enough spatial margin to absorb the remaining updates. This explains why the all-time bound~\eqref{eq:geomhalfadvance} follows even when different directions use different block lengths.
\begin{proof}
Fix $\mathbf{e},w$ and a covering index $j$, and abbreviate $\F=\F_{\omega_{w,\mathbf{e}}}$. The global bound $\F^n\mathbf{h}\leq \mathbf{b}_{\max}$ follows from $T(\mathbf{b}_{\max})\leq \mathbf{b}_{\max}$. Thus~\eqref{eq:geomdiscreteadvance} implies $\F^{N_j}\mathbf{h}(s)\leq \mathbf{h}(s-2)$. Order preservation and translation invariance give, for every integer $k\geq1$,
\begin{equation}\label{eq:geomblockadvance}
 \F^{kN_j}\mathbf h(s)\leq\F^{N_j}\mathbf h\bigl(s-2(k-1)\bigr)
 \leq\boldsymbol\alpha/2\qquad(s\leq2k).
\end{equation}
For the first inequality, iterate $\F^{N_j}\mathbf h\leq\mathbf h(\,\cdot-2)$; the second is the discrete advance bound~\eqref{eq:geomdiscreteadvance} at the shifted argument.

Write $t=kN_j+\tau$ with $0\leq \tau<N_j$. Each update has dependence radius $2\rho$. The additional $\tau$ updates therefore preserve the bound $\boldsymbol{\alpha}/2$ on $s\leq2k-2\rho \tau$, because~\eqref{eq:smallbox} gives $T^\tau(\boldsymbol{\alpha}/2)\leq\boldsymbol{\alpha}/2$. For $t\geq t_*$, we have $k\geq\lceil\rho N_*+1\rceil$ and $\tau<N_*$. Hence $2k-2\rho \tau\geq2$, which proves~\eqref{eq:geomhalfadvance}. This residual-time argument does not require the free step trajectory to decrease at every iteration.
\end{proof}

\subsection{Expansion of a large ball}\label{sec:ball}
This is the geometric passage to Step 3 of Figure~\ref{fig:proofstory}. The advancing balls are comparison profiles with small positive messages, while the actual shortened set remains fixed. Figure~\ref{fig:geometry} details the local plane comparison underlying this expansion.

Write $B_R(\mathbf{x})=\{\mathbf{z}\in\R^d:\norm{\mathbf{z}-\mathbf{x}}_2\leq R\}$ for the closed Euclidean ball and $B_R=B_R(0)$. A finite number of iterations depends only on a bounded neighborhood of each position. Within that neighborhood, a sufficiently large ball differs little enough from its tangent half-space to apply the uniform half-space advancement of Lemma~\ref{lem:commonhalfspace}.

\begin{figure}[htbp]
\centering

\def\GeomLeft{(a) Local half-space comparison}
\def\GeomRight{(b) Expansion of the comparison profile}

\begin{tikzpicture}[x=1cm,y=1cm,>=stealth,font=\small]
\begin{scope}[scale=.92]
 \clip (-2.35,-2.45) rectangle (2.5,2.6);
 \fill[blue!10] (-2.3,-2.2) -- plot[domain=-2.2:2.2,samples=80] ({sqrt(100-\x*\x)-10},\x) -- (-2.3,2.2) -- cycle;
 \begin{scope}
  \clip (.3,0) circle (2);
  \fill[orange!30] (-2.3,-2.2) rectangle (-1,2.2);
 \end{scope}
 \draw[blue!65!black,thick] plot[domain=-2.2:2.2,samples=80] ({sqrt(100-\x*\x)-10},\x);
 \draw[gray,densely dotted] (0,-2.2)--(0,2.2);
 \draw[orange!70!black,dashed,thick] (-1,-2.2)--(-1,2.2);
 \draw[black,dashed] (.3,0) circle (2);
 \fill (.3,0) circle (.045) node[above right] {$\mathbf{x}$};
 \draw[->] (.3,0)--(1.45,0) node[midway,below] {$\mathbf{e}$};
 \draw[->] (.3,0)--(1.3,1.732) node[midway,right] {$J$};
 \node[anchor=east,align=right] at (-1.15,-.7) {$\mathbf{H}=\boldsymbol{\alpha}$};
 \node[blue!65!black,anchor=east] at (-.25,1.75) {$\partial B_{R_0}$};
 \node at (-1,2.42) {$R_0-1$};
 \node at (0,2.42) {$R_0$};
\end{scope}
\node[anchor=north,align=center] at (0,-3.35) {\GeomLeft};
\begin{scope}[shift={(6.6,0)}]
 \fill[blue!10] (0,0) circle (1.5);
 \draw[blue!65!black,thick] (0,0) circle (1.5);
 \draw[green!40!black,thick,dash dot] (0,0) circle (1.9);
 \fill (0,0) circle (.035);
 \draw[->] (0,0)--(-1.25,.83) node[midway,above] {$R_0$};
 \draw[->] (0,0)--(1.65,.94) node[pos=.62,above,fill=white,inner sep=1pt] {$R_0+1$};
 \node at (0,-.4) {$t=0:\quad \mathbf{v}_{R_0}=\boldsymbol{\alpha}$};
 \node[anchor=north,align=center] at (0,-2.1) {$\norm{\mathbf{z}}_2\leq R_0+1:$\\[2pt]
   $(\F_{\nu_w}^{t_*}\mathbf{v}_{R_0})(\mathbf{z})\leq\boldsymbol{\alpha}/2$};
 \node[anchor=north,align=center] at (0,-3.35) {\GeomRight};
\end{scope}
\end{tikzpicture}

\caption{Two-dimensional sections of the geometric comparison. In (a), the dashed circle is the dependence ball $B_J(\mathbf{x})$. Its orange part, defined by $\mathbf{e}\cdot \mathbf{z}\leq R_0-1$, lies inside the initial ball $B_{R_0}$ by~\eqref{eq:curvature}. The dotted line is the tangent plane $\mathbf{e}\cdot \mathbf{z}=R_0$, and $\mathbf{H}(\mathbf{z})=\mathbf{h}(\mathbf{e}\cdot \mathbf{z}-R_0+1)$. In (b), the comparison profile $\mathbf{v}_{R_0}$ initially equals $\boldsymbol{\alpha}$ on the blue ball and $\mathbf{b}_{\max}$ outside it; after $t_*$ updates it is at most $\boldsymbol{\alpha}/2$ on the larger ball. The actual shortened set $S$, omitted from the drawing, contains the initial ball on the lattice and has zero initial messages. Its initial profile is bounded above by the comparison profile, as used in Section~\ref{sec:torus}. Lengths and radius ratios are schematic; the radial increase is exaggerated for visibility.}\label{fig:geometry}

\end{figure}
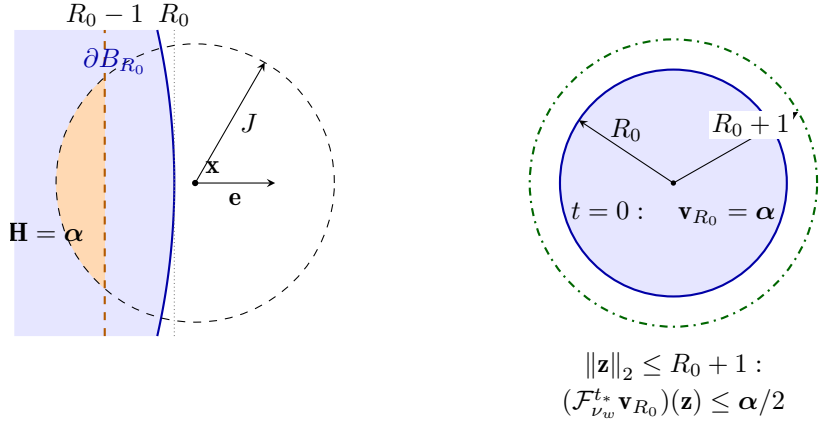

Define the dependence range, radius, and comparison profile by

\begin{equation}\label{eq:geomradius}
 J=2\rho t_*,\qquad R_0=2J+J^2+1,\qquad
 \mathbf{v}_R(\mathbf{x})=\begin{cases}\boldsymbol{\alpha},&\norm{\mathbf{x}}_2\leq R,\\ \mathbf{b}_{\max},&\norm{\mathbf{x}}_2>R.\end{cases}
\end{equation}

\begin{lemma}[Ball expansion]\label{lem:ballexpansion}
For every $R\geq R_0$ and $w\geq w_0$,

\begin{equation}\label{eq:geomballadvance}
 \F_{\nu_w}^{t_*}\mathbf{v}_R(\mathbf{x})\leq\boldsymbol{\alpha}/2\quad(\norm{\mathbf{x}}_2\leq R+1),
 \qquad \F_{\nu_w}^{t_*}\mathbf{v}_R\leq \mathbf{b}_{\max}\quad\text{everywhere}.
\end{equation}

\end{lemma}

Here $J=2\rho t_*$ in~\eqref{eq:geomradius} is the distance within which initial messages can affect a point after the common time $t_*$. The initial comparison profile $\mathbf v_R$ in the same equation assigns the small positive bound $\boldsymbol\alpha$ to a ball of radius $R$. The expansion bound~\eqref{eq:geomballadvance} says that, after $t_*$ updates, the smaller bound $\boldsymbol\alpha/2$ holds on the larger ball of radius $R+1$. To obtain this conclusion, the curvature estimate~\eqref{eq:curvature} only needs to compare a sphere and a plane inside the dependence ball of radius $J$. The shortfall from the tangent plane is absorbed by shifting the planar step inward by one unit in~\eqref{eq:geomballcomparison}.
\begin{proof}
The update at $\mathbf{x}$ depends only on initial values within distance $J$. If $\norm{\mathbf{x}}_2<R/2$, this region lies in $B_R$ because $R>2J$. Its value is $T^{t_*}(\boldsymbol{\alpha})\leq\theta^{t_*}\boldsymbol{\alpha}\leq\boldsymbol{\alpha}/2$, by~\eqref{eq:geomdirectionaltime}.

For $R/2\leq R_x:=\norm{\mathbf{x}}_2\leq R+1$, put $\mathbf{e}=\mathbf{x}/R_x$. For any displacement $\norm{\boldsymbol{\delta}}_2\leq J$, Taylor's formula for the Euclidean norm gives

\begin{equation}\label{eq:curvature}
 0\leq\norm{\mathbf{x}+\boldsymbol{\delta}}_2-R_x-\mathbf{e}\cdot \boldsymbol{\delta}
 \leq\frac{J^2}{2(R_x-J)}\leq\frac{J^2}{R-2J}<1.
\end{equation}

Indeed, the Hessian of the norm at $\mathbf{y}\ne0$ is $(I-\mathbf{u}\mathbf{u}^{\mathsf T})/\norm{\mathbf{y}}_2$, with $\mathbf{u}=\mathbf{y}/\norm{\mathbf{y}}_2$, and has operator norm at most $1/\norm{\mathbf{y}}_2$. Along $\mathbf{y}=\mathbf{x}+t\boldsymbol{\delta}$, $0\leq t\leq1$, the denominator is at least $R_x-J>0$. The integral remainder gives the factor $1/2$, and $R\geq2J+J^2+1$ gives the strict final inequality in~\eqref{eq:curvature}. This also covers $d=1$, where the Hessian is zero away from the origin.

Set $\mathbf{H}(\mathbf{z})=\mathbf{h}(\mathbf{e}\cdot \mathbf{z}-R+1)$. If $\mathbf{z}=\mathbf{x}+\boldsymbol{\delta}$ is in the dependence range and $\mathbf{e}\cdot \mathbf{z}\leq R-1$, the curvature bound~\eqref{eq:curvature} gives $\norm{\mathbf{z}}_2<R$. Thus $\mathbf{v}_R(\mathbf{z})\leq \mathbf{H}(\mathbf{z})$ throughout the dependence range: where $\mathbf{H}=\boldsymbol{\alpha}$, both initial values equal $\boldsymbol{\alpha}$; elsewhere $\mathbf{H}=\mathbf{b}_{\max}\geq \mathbf{v}_R$. A planar profile updates with the projected one-dimensional distribution. Finite dependence and order therefore give

\begin{equation}\label{eq:geomballcomparison}
 \F_{\nu_w}^{t_*}\mathbf{v}_R(\mathbf{x})
 \leq\F_{\omega_{w,\mathbf{e}}}^{t_*}\mathbf{h}(R_x-R+1)\leq\boldsymbol{\alpha}/2,
\end{equation}

where the last inequality uses $R_x-R+1\leq2$ and the uniform half-space bound~\eqref{eq:geomhalfadvance}. The global bound follows from $\mathbf{v}_R\leq \mathbf{b}_{\max}$ and $T(\mathbf{b}_{\max})\leq \mathbf{b}_{\max}$.
\end{proof}

\begin{remark}[Critical nuclei and the size of the initial region]\label{rem:nucleation}
Classical nucleation theory~\cite{Volmer1926} explains why a small favorable region can fail to grow: bulk free-energy gain competes with interfacial cost. In the isotropic short-range approximation used for comparison in~\cite[Fig. 1 and accompanying text]{Nishino2011}, a spherical droplet in dimension $d\geq2$ has
\begin{equation}\label{eq:nucleation}
 \mathcal G(R)=d\omega_d\sigma R^{d-1}-\omega_d\Delta f\,R^d,
 \qquad R_{\mathrm c}=\frac{(d-1)\sigma}{\Delta f}.
\end{equation}
Here $\omega_d$ is the unit-ball volume, $\sigma>0$ the interfacial free energy per unit area, and $\Delta f>0$ the bulk gain per unit volume. The maximum occurs at $R_{\mathrm c}$; growth above it lowers $\mathcal G$. At fixed $\sigma,\Delta f$, a fixed initial radius can become insufficient as $d$ increases, although $R_{\mathrm c}$ remains finite at each finite $d$. Favorable initial regions and propagation costs also arise in coupled constraint problems and mean-field models~\cite{Hassani2013,Caltagirone2014}.

For decoding, it is the small-message region surrounding the permanently shortened set that grows. We do not identify $\sigma,\Delta f$ with the recursion's potential. The rigorous comparison is~\eqref{eq:curvature}: after fixing the common finite time, a sufficiently large radius makes curvature small in the whole dependency region. The sufficient radius $R_0$ in~\eqref{eq:geomradius} works for every $R\geq R_0$; it is not a minimal radius. Section~\ref{sec:numerics} illustrates dependence on the initial square size without testing the physical formula~\eqref{eq:nucleation}.
\end{remark}

\subsection{Translation to a finite torus}\label{sec:torus}
Ball expansion supplies a local comparison. Its translations establish Step 4 of Figure~\ref{fig:proofstory}: all positions become small, after which contraction gives complete decoding.

Because $\nu_w$ is supported on $(1/w)\mathbb Z^d$, every update of $\F_{\nu_w}$ stays within a coset of this lattice. On the zero coset, the scale relation~\eqref{eq:scalerelation} identifies this update with the integer-lattice operator $\F_w$ in~\eqref{eq:latticeoperator}. We write $\F=\F_w$ below. Thus the preceding comparison on $\R^d$ applies at $\mathbf x=\mathbf i/w$ to the actual offsets $\mathbf k\in K_w$.

On the integer lattice define
\begin{equation}\label{eq:geomlatticeball}
 \mathcal B=\{\mathbf{i}\in\mathbb Z^d:\norm{\mathbf{i}/w}_2\leq R_0\},\qquad
 \mathcal E=\{-1,0,1\}^d.
\end{equation}

Let $\mathbf{v}_{\mathcal B}$ equal $\boldsymbol{\alpha}$ on $\mathcal B$ and $\mathbf{b}_{\max}$ elsewhere. Since $w\geq\lceil\rho\rceil$, every displacement in $\mathcal E$ has scaled length at most one. Lemma~\ref{lem:ballexpansion} (ball expansion) therefore gives

\begin{equation}\label{eq:geomblockexpansion}
 \F^{t_*}\mathbf{v}_{\mathcal B}\leq \mathbf{v}_{\mathcal B+\mathcal E}.
\end{equation}

Here $\F$ is the unshortened lattice operator, and the profile on the right again uses the two levels $\boldsymbol{\alpha},\mathbf{b}_{\max}$.

Shorten a hypercube containing $\mathcal B$ in the torus $(\mathbb Z/L\mathbb Z)^d$. Let $\mathbf{z}_t$ denote the periodic lift of the actual density evolution, initially zero on the shortened positions and at most $\mathbf{b}_{\max}$ elsewhere. Then, for every $m\geq0$ and every $\boldsymbol{\delta}\in\{-m,\ldots,m\}^d$,

\begin{equation}\label{eq:geomtranslate}
 \mathbf{z}_{mt_*}\leq\tau_{\boldsymbol{\delta}} \mathbf{v}_{\mathcal B},
 \qquad (\tau_{\boldsymbol{\delta}} \mathbf{v})(\mathbf{i})=\mathbf{v}(\mathbf{i}-\boldsymbol{\delta}).
\end{equation}

To prove this, initially $\mathbf{z}_0\leq \mathbf{v}_{\mathcal B}$; additional periodic copies of the shortened region only reduce $\mathbf{z}_0$. If~\eqref{eq:geomtranslate} holds at $m$ and $\boldsymbol{\delta}$, then for every $\mathbf{e}\in\mathcal E$,

\begin{equation}\label{eq:geomtranslationinduction}
 \mathbf{z}_{(m+1)t_*}\leq\F^{t_*}\mathbf{z}_{mt_*}
 \leq\tau_{\boldsymbol{\delta}}\F^{t_*}\mathbf{v}_{\mathcal B}
 \leq\tau_{\boldsymbol{\delta}} \mathbf{v}_{\mathcal B+\mathcal E}
 \leq\tau_{\boldsymbol{\delta}+\mathbf{e}}\mathbf{v}_{\mathcal B}.
\end{equation}

The inequalities use, in order, shortening, monotonicity and translation invariance,~\eqref{eq:geomblockexpansion}, and $\mathcal B+\mathbf{e}\subset\mathcal B+\mathcal E$. This proves the induction. A separate inequality is maintained for each translate; no commutation with componentwise minima is assumed.

Use $s(w)$ defined in~\eqref{eq:seedside}, with

\begin{equation}\label{eq:geomseed}
 L>s(w)+2w.
\end{equation}

The hypercube contains $\mathcal B$ and fits in the torus; the coupling offsets are distinct modulo $L$. In~\eqref{eq:geomtranslate}, take $m=L$ and, for each representative $\mathbf{i}\in\{0,\ldots,L-1\}^d$, choose $\boldsymbol{\delta}=\mathbf{i}$. Since $0\in\mathcal B$, every position satisfies $\mathbf{z}_{Lt_*}(\mathbf{i})\leq\boldsymbol{\alpha}$. The small-message contraction condition~\eqref{eq:smallbox} then gives uniform convergence to zero. This completes the discrete local-shortening conclusion of Theorem~\ref{thm:transfer} (transfer to local shortening).

The translated-profile bound~\eqref{eq:geomtranslate} is a family of upper bounds for the same actual profile $\mathbf z_{mt_*}$. Its translation $\tau_{\boldsymbol\delta}$ centers the small-message ball at $\boldsymbol\delta$. As the block count $m$ increases, more centers are available. To bound one target site, choose a center whose ball contains it; different target sites may use different comparisons. The shortened set $S$ from Definition~\ref{def:DE} (shortened density evolution) stays fixed throughout this argument.

\subsection{What the bounds make explicit}\label{sec:quantitative}
The finite collection of directional advancement times $N_j$ determines the sufficient coupling width $w_0$ in~\eqref{eq:geomwidth}, the common advancement time $t_*$ in~\eqref{eq:geomcommontime}, and the initial comparison radius $R_0$ in~\eqref{eq:geomradius}. In particular,

\begin{equation}\label{eq:geomsizebound}
 R_0=4\rho t_*+4\rho^2t_*^2+1,
 \qquad s(w)\leq2wR_0+3.
\end{equation}

For a prescribed shortening ratio $\zeta>0$, the bound

\begin{equation}\label{eq:geomratetolerance}
 L\geq s(w)(1+1/\zeta)^{1/d},\qquad L>s(w)+2w
 \quad\Longrightarrow\quad
 \frac{s(w)^d}{L^d-s(w)^d}\leq\zeta
\end{equation}

controls the ratio of shortened to remaining positions. The corresponding rate bounds depend on the ensemble's variable and check counts. For a uniform componentwise message target $\xi\boldsymbol{\alpha}$, $0<\xi<1$, a sufficient iteration count is

\begin{equation}\label{eq:geomiterationbound}
 Lt_*+\left\lceil\frac{\log\xi}{\log\theta}\right\rceil.
\end{equation}

For fixed dimension, channel, degrees, and chosen coupling width, this is linear in the torus side $L$ for a fixed target.

The proof establishes that the $N_j$ are finite, but does not yet bound them numerically from the degrees and channel parameter alone. Thus these sufficient formulas expose the remaining finite-time quantities without asserting practical shortening sizes. All constants may depend on the fixed dimension; the argument does not give uniform bounds when the dimension grows.

The dependence on the largest selected advancement time is particularly costly. Since $K\geq1$, the width choice~\eqref{eq:geomwidth}, common-time definition~\eqref{eq:geomcommontime}, and size bound~\eqref{eq:geomsizebound} imply, for fixed $d$,
\begin{align}\label{eq:constantdependence}
 w_0&=O_d(\alpha_{\min}^{-1}N_*K^{N_*}),& t_*&=O_d(N_*^2),\notag\\
 R_0&=O_d(N_*^4),&s(w_0)&=O_d(\alpha_{\min}^{-1}N_*^5K^{N_*}).
\end{align}
Here $O_d$ hides constants depending only on $d$; $K$ and $\alpha_{\min}$ remain displayed. These are sufficient upper bounds expressed through $N_*$, not lower bounds on the required sizes. In particular, the exponential factor comes from finite-time perturbation control. The proof supplies no quantitative dependence of $N_*$ on the gap to the threshold. The finite directional cover also depends on $d$, so these formulas give no estimate uniform in dimension. Once that cover and its times have been chosen, $R_0$ is independent of every subsequently chosen $w\geq w_0$.

For fixed degrees, dimension, width, and message target, a direct density-evolution update costs $O(L^d)$ arithmetic operations. Thus~\eqref{eq:geomiterationbound} gives an $O(L^{d+1})$ sufficient operation bound for this synchronous calculation. A finite-code complexity estimate would additionally require the section size chosen in Lemma~\ref{lem:finiteDE} (finite-iteration approximation).

\begin{remark}[Channel dependence of the sufficient choices]\label{rem:channelconstants}
The choices of the small rectangle in~\eqref{eq:regularbox} and~\eqref{eq:mng2box} shrink as $\eps\downarrow0$, and substituting them separately into~\eqref{eq:geomwidth} can worsen that bound. A single construction chosen at any fixed $\eps_0$ below the relevant threshold works for all $0\leq\eps\leq\eps_0$: couple the BEC observations so that smaller erasure probabilities reveal more bits, and use order preservation of every update. Thus no increase in width or shortening is needed when improving the channel. At the opposite endpoint, for fixed MN degrees with $g=2$,
\begin{equation}\label{eq:thetalimit}
 \lim_{\eps\uparrow1-r/\ell}\frac{1+\eps}{2}
 =1-\frac{r}{2\ell}<1.
\end{equation}
The contraction factor in~\eqref{eq:mng2box} therefore does not itself approach one at the capacity boundary. Any deterioration of the advancement-time estimates near that boundary requires a separate analysis of $N_*$; no divergence rate follows from the present proof.
\end{remark}

\section{Regular LDPC codes}\label{sec:regular}
For regular LDPC codes, we verify Step 1 of Figure~\ref{fig:proofstory}: the potential-threshold definition gives positivity, and a union bound gives contraction near zero. The common spatial proof then supplies Steps 2--4 and establishes Theorem~\ref{thm:main} (regular-code decoding) below the potential threshold. Here $m=1$, so one-component vectors are identified with scalars and written in ordinary type. No classification of the positive scalar fixed points is needed.

At each position place $M$ variables of degree $d_v$ and $d_vM/d_c$ checks of degree $d_c$, with socket allocation as in Definition~\ref{def:ensemble} (socket allocation and shortening). Put
\begin{equation}\label{eq:scalarfunctions}
 p=d_v-1\geq2,\quad q=d_c-1,\quad
 Q(x)=1-(1-x)^q,\quad f(y)=\eps y^p,\quad b_{\max}=\eps,\quad D=1.
\end{equation}
These are the BEC density-evolution functions~\cite{RU2001}; substituting them in the integer-lattice update~\eqref{eq:latticeoperator} gives the regular-code recursion. For $\eps>0$, they satisfy the map and primitive assumptions with
\begin{equation}\label{eq:regularprimitives}
 G(x)=x-\frac{1-(1-x)^{q+1}}{q+1},\qquad
 F(y)=\frac{\eps}{p+1}y^{p+1}.
\end{equation}

\begin{definition}[Regular-code potential threshold]\label{def:potential}
The scalar potential~\cite{Maxwell2014} and its threshold are
\begin{align}
 U(x;\eps)&=xQ(x)-G(x)-\frac{\eps}{p+1}Q(x)^{p+1},\label{eq:potential}\\
 \pot&=\sup\{\eps\in[0,1]:U(x;\eps)\geq0\text{ for all }x\in[0,1]\}.\label{eq:potthreshold}
\end{align}
\end{definition}

If $0<\eps<\pot$, choose $\eps<\eps'<\pot$. The admissible parameters in~\eqref{eq:potthreshold} form a downward-closed set, so $U(x;\eps')\geq0$ for every $x$. Therefore
\begin{equation}\label{eq:regularpositive}
 U(x;\eps)=U(x;\eps')+
 \frac{\eps'-\eps}{p+1}Q(x)^{p+1}>0
 \qquad(0<x\leq1).
\end{equation}
In particular, the fixed-point positivity condition~\eqref{eq:positivefixed} holds. For the small interval, the union bound $Q(x)\leq qx$ gives a sufficient choice
\begin{equation}\label{eq:regularbox}
 \beta=q^{-p/(p-1)},\qquad
 \alpha=\min\{\eps/2,\beta/4\},\qquad \theta=\tfrac14.
\end{equation}
Indeed, for $0\leq t\leq1$,
\begin{equation}\label{eq:regularcontraction}
 T(t\alpha)\leq q^p(t\alpha)^p
 \leq(\alpha/\beta)^{p-1}t\alpha\leq\tfrac14t\alpha.
\end{equation}
Theorem~\ref{thm:transfer} (transfer to local shortening) now proves message convergence in Theorem~\ref{thm:main} (regular-code decoding). For a variable at the integer position $\mathbf i$, write $y(\mathbf i)=A_wQ(A_wu)(\mathbf i)$ using the lattice average~\eqref{eq:latticeoperator}. On the computation tree, the $d_v$ incoming check branches and the channel observation are independent~\cite[Section III]{MD2013}. Thus, outside the shortened set,
\begin{equation}\label{eq:regularposterior}
 e_{\mathrm{post}}(\mathbf i)=\eps y(\mathbf i)^{d_v}
 \leq\eps(qu_0)^{d_v}\quad\text{if }0\leq u(\mathbf i)\leq u_0\text{ at all positions}.
\end{equation}
The inequality follows from $Q(x)\leq qx$ and averaging. Message convergence therefore implies posterior convergence. When $\eps=0$, all regular-code messages are zero from the start.

\subsection{Finite codes and rate}\label{sec:rate}
Fix finite $L,w$. For an arbitrarily small target posterior erasure probability, first choose the finite iteration count guaranteed by Theorem~\ref{thm:main} (regular-code decoding), and then take $M$ large in Lemma~\ref{lem:finiteDE} (finite-iteration approximation). This gives finite graph realizations with a small average bit-erasure probability.

Write $s=s(w)$ for the shortened side length defined in~\eqref{eq:seedside}. There are $s^d$ shortened positions and $n=M(L^d-s^d)$ transmitted variables; let $k$ be the dimension of the shortened code. The check-count argument for local shortening~\cite[Section V]{MD2013} gives
\begin{equation}\label{eq:rate}
 \frac{k}{n}\geq
 1-\frac{d_v}{d_c}\frac{L^d}{L^d-s^d}
 =1-\frac{d_v}{d_c}-O(L^{-d}).
\end{equation}
The inequality allows dependent checks; an equality for the actual rate would need a rank argument. To compare geometries at the same dimension, let $V=L^d$ and consider a shortened slab of thickness $h$ and a shortened cube of side $s$. Their exact shortening fractions are
\begin{equation}\label{eq:shorteningcomparison}
 \frac{|S_{\mathrm{slab}}|}{V}=\frac{hL^{d-1}}{L^d}=hV^{-1/d},
 \qquad \frac{|S_{\mathrm{cube}}|}{V}=\frac{s^d}{L^d}=s^dV^{-1}.
\end{equation}
Fix the degrees, channel, dimension, and a common coupling width $w$ large enough for both constructions; then choose finite $h$ and $s$ independently of $L$. The slab recursion reduces to the one-dimensional recursion~\cite[Proposition 1]{MD2013}. For regular codes, the corresponding reductions from $1-d_v/d_c$ in the check-count rate bound are $(d_v/d_c)h/(L-h)$ and $(d_v/d_c)s^d/(L^d-s^d)$. Their orders are $V^{-1/d}$ and $V^{-1}$, respectively, so the exponent improves for $d>1$. The sufficient $w$ and $s$ for the cube can be large: this is an asymptotic comparison at fixed sizes, not a finite-length ranking. For MN codes the geometric term in~\eqref{eq:mnactualrate} has the same dependence on the shortening fraction, with the punctured nullity controlled separately in Section~\ref{sec:mnrate}.

\subsection{Numerical illustration in two dimensions}\label{sec:numerics}
This example illustrates Step 1 and the dependence of propagation on the initial shortened region. It is separate from the analytic proof of Theorem~\ref{thm:main} (regular-code decoding). Figure~\ref{fig:regularpotential} plots the potential~\eqref{eq:potential} for regular $(3,6)$ codes below and above $\pot\simeq0.48815$. Unlike the schematic in Figure~\ref{fig:proofstory}, these curves show the scalar potential itself; the MN potential has two arguments.

\begin{figure}[tb]
\centering
\input{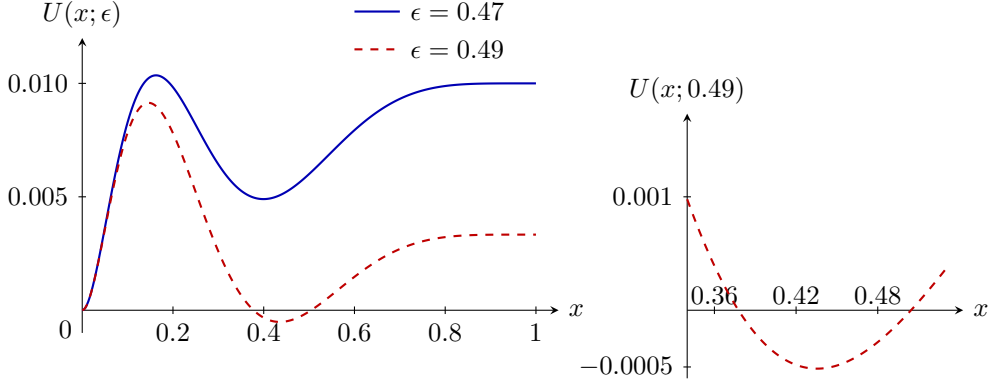}
\caption{Scalar potential~\eqref{eq:potential} for regular $(3,6)$ codes. At $\eps=0.47$ it is positive for $x>0$; at $\eps=0.49$ it is negative on part of the interval. The right panel magnifies the region of negative potential for $\eps=0.49$. The curves are numerical evaluations; positivity below the threshold follows analytically from~\eqref{eq:regularpositive}.}\label{fig:regularpotential}
\end{figure}

We iterate the exact integer-lattice recursion~\eqref{eq:latticeoperator} with the scalar functions~\eqref{eq:scalarfunctions}, dimension $d=2$, $\eps=0.47$, and $w=2$. Each average is uniform over the $5\times5$ offsets. On an $L\times L$ torus, initialize $u=\eps$ outside a centered shortened square of odd side $s$ and hold $u=0$ inside it at every iteration. All positions update synchronously. Without shortening, the same initialization converges numerically to the positive homogeneous value $0.3994225867$, so this example lies in a regime where shortening changes the outcome.

Table~\ref{tab:numerics} reports floating-point results at two torus sizes. We stop when either $\max_{\mathbf i}u_t(\mathbf i)<10^{-10}$ (criterion A), or $t\geq100$ and $\max_{\mathbf i}|u_t(\mathbf i)-u_{t-1}(\mathbf i)|<10^{-13}$ (criterion B), with a limit of $10^4$ iterations. The calculations use double precision; the separable implementation of the two-dimensional average was checked against direct summation over all offsets and against the homogeneous recursion.

\begin{table}[tb]
\centering\small
\caption{Two-dimensional regular $(3,6)$ density evolution with $\eps=0.47$, $w=2$, and a shortened $s\times s$ square. A denotes a small maximum message; B denotes a small update residual. A small residual is numerical evidence of a plateau, not a proof of nonconvergence.}\label{tab:numerics}
\begin{tabular}{rrrrr}\toprule
$L$ & $s$ & Iterations & Criterion & $\max_{\mathbf i}u_t(\mathbf i)$\\\midrule
65 & 3 & 100 & B & $0.3994225867$\\
65 & 7 & 156 & B & $0.3994225866$\\
65 & 11 & 305 & B & $0.3994225855$\\
65 & 15 & 904 & B & $0.3994225769$\\
65 & 17 & 1020 & A & $2.14\times10^{-15}$\\
65 & 19 & 519 & A & $1.60\times10^{-18}$\\
129 & 3 & 100 & B & $0.3994225867$\\
129 & 7 & 156 & B & $0.3994225867$\\
129 & 11 & 305 & B & $0.3994225867$\\
129 & 15 & 904 & B & $0.3994225867$\\
129 & 17 & 1450 & A & $1.19\times10^{-15}$\\
129 & 19 & 949 & A & $7.61\times10^{-19}$\\
\bottomrule\end{tabular}

\end{table}

For both values of $L$, squares of side at most $15$ among those tested reach a positive plateau, whereas sides $17$ and $19$ reach criterion A. This behavior is consistent with the size requirement discussed in Remark~\ref{rem:nucleation}. The experiment neither determines a rigorous critical side nor bounds the constants $N_*,w_0,R_0$ of Section~\ref{sec:quantitative}. In particular, testing $w=2$ does not certify the theorem's sufficient width, and these density-evolution values are not finite-graph error measurements. The reproduction script \texttt{anc/regular\_bec.py} and its numeric outputs accompany the source files.

\section{MacKay--Neal codes at general fixed degrees}\label{sec:mn}
For MN codes, Step 1 in Figure~\ref{fig:proofstory} holds below the BEC capacity boundary for all integers $\ell>r\geq2$ and $g\geq2$. We prove small-message contraction and classify the fixed points before evaluating their potentials. An elementary comparison reduces nontrivial fixed-point positivity to $r=g=2$. Steps 2--4 then give local-shortening decoding, and Section~\ref{sec:mnrate} adds the transmitted-rate argument needed for Step 5, capacity achievement.

\subsection{Ensemble, update, and a small invariant region}\label{sec:mnmodel}
Use the spatially coupled MN construction~\cite{MN2011}. At each position of $(\Z/L\Z)^d$, place $rM/\ell$ punctured variables of degree $\ell$, $M$ transmitted variables of degree $g$, and $M$ checks with $r$ punctured and $g$ transmitted sockets. Partition each variable-socket type independently into the $(2w+1)^d$ offset groups indexed by $\{-w,\ldots,w\}^d$, and match the corresponding variable and check groups uniformly. Require $rM/\ell$, $rM/(2w+1)^d$, and $gM/(2w+1)^d$ to be integers. In the shortened hypercube, fix both variable types to zero and omit their transmissions.

The two-message BEC recursion is given in~\cite{MN2011,MN2014}. For $0<\eps<1$, write $\mathbf{u}=(a,b)^{\mathsf T}$ for the punctured and transmitted variable-to-check erasure profiles and set
\begin{align}
 Q_1(a,b)&=1-(1-a)^{r-1}(1-b)^g,\notag\\
 Q_2(a,b)&=1-(1-a)^r(1-b)^{g-1},\label{eq:mncheck}\\
 f(\mathbf{y})&=(y_1^{\ell-1},\eps y_2^{g-1})^{\mathsf T},\qquad
 \mathbf{b}_{\max}=(1,\eps)^{\mathsf T},\label{eq:mnf}\\
 \F_\lambda \mathbf{u}&=f\bigl(A_\lambda Q(A_\lambda \mathbf{u})\bigr).
 \label{eq:mnoperator}
\end{align}
For the finite code ensemble, the averaging operator in the MN recursion~\eqref{eq:mnoperator} is the integer-lattice average $A_w$ from~\eqref{eq:latticeoperator}; the scale relation~\eqref{eq:scalerelation} gives the equivalent formulation with $A_{\nu_w}$. The averaging order follows from independent socket matching. Both maps are order preserving, and $T(\mathbf{b}_{\max})=\mathbf{b}_{\max}$ for $T=f\circ Q$. For $\mathbf y=A_\lambda Q(A_\lambda\mathbf u)$, independence of the incoming branches in the computation tree~\cite{MN2011} gives, at an unshortened position,
\begin{equation}\label{eq:mnposterior}
 e_{P,\mathrm{post}}=y_1^\ell,\qquad e_{T,\mathrm{post}}=\eps y_2^g.
\end{equation}
All $\ell$ branches are needed to leave a punctured variable erased; for a transmitted variable, all $g$ branches and its channel observation must be erased. Continuity of the averages and $Q(0)=0$ give posterior convergence from message convergence. On the unit square, the check update $Q$ has Lipschitz constant at most $h_0=r+g-1$, and the variable update $f$ has Lipschitz constant at most $\max\{\ell-1,g-1\}$ in the componentwise supremum norm. Thus one full update is Lipschitz with constant
\begin{equation}\label{eq:mnconstants}
 K_{\mathrm{MN}}=(r+g-1)\max\{\ell-1,g-1\}.
\end{equation}

\begin{lemma}[A contracting rectangle]\label{lem:mnbox}
For every $\ell>r\geq2$, $g\geq2$, and $0<\eps<1$, there exist $0<\boldsymbol{\alpha}\leq\mathbf{b}_{\max}$ and $0<\theta<1$ such that the update~\eqref{eq:mnoperator} satisfies the small-box condition~\eqref{eq:smallbox} for every probability averaging measure $\lambda$.
\end{lemma}
\begin{proof}
If $g\geq3$, define
\begin{equation}\label{eq:mng3box}
 p_{\min}=\min\{\ell-1,g-1\},\quad
 \beta=h_0^{-p_{\min}/(p_{\min}-1)},\quad
 \eta=\min\{\eps/2,\beta/4\},\quad
 \boldsymbol{\alpha}=(\eta,\eta)^{\mathsf T},\quad\theta=\tfrac14.
\end{equation}
For $0\leq t\leq1$ and $0\leq \mathbf{u}\leq t\boldsymbol{\alpha}$, the union bound in~\eqref{eq:mncheck} gives $Q_1,Q_2\leq h_0t\eta\leq1$. Averaging preserves this bound, so each updated component is at most
\begin{equation}\label{eq:mng3contraction}
 h_0^{p_{\min}}(t\eta)^{p_{\min}}
 =t\eta\left(\frac{t\eta}{\beta}\right)^{p_{\min}-1}
 \leq\tfrac14t\eta.
\end{equation}

For $g=2$, the transmitted update is linear to first order near zero. We allow the transmitted bound to be larger than the punctured bound. Define
\begin{align}
 \theta&=\frac{1+\eps}{2},\qquad
 c=\max\left\{1,\frac{2\eps r}{1-\eps}\right\},\qquad
 H_1=r-1+2c,\notag\\
 \eta&=\min\left\{\frac12,\frac{\eps}{2c},
 \left(\frac{\theta}{H_1^{\ell-1}}\right)^{1/(\ell-2)}\right\},
 \qquad\boldsymbol{\alpha}=\eta(1,c)^{\mathsf T}.
 \label{eq:mng2box}
\end{align}

The ratio $c$ in the rectangle choice~\eqref{eq:mng2box} allows a larger transmitted-message bound than the punctured-message bound. The margin $\theta-\eps>0$ absorbs the contribution of the punctured coordinate to the transmitted update. The scale $\eta$ then makes the punctured update, whose exponent $\ell-1$ is at least two, contract as well. The component bounds~\eqref{eq:mng2contraction} verify these two requirements separately.

For $0\leq t\leq1$ and $0\leq\mathbf u\leq t\boldsymbol\alpha$, the same union bound gives $Q_1\leq H_1t\eta$ and $Q_2\leq(r+c)t\eta$. Hence
\begin{align}
 (\F_\lambda \mathbf{u})_1&\leq(H_1t\eta)^{\ell-1}\leq\theta t\eta,\notag\\
 (\F_\lambda \mathbf{u})_2&\leq\eps(r+c)t\eta\leq\theta ct\eta.
 \label{eq:mng2contraction}
\end{align}
The first inequality uses the small scale $\eta$ in the rectangle choice~\eqref{eq:mng2box}; the last uses $c(\theta-\eps)\geq\eps r$. Thus $\F_\lambda \mathbf{u}\leq\theta t\boldsymbol{\alpha}$ in both cases. Iteration gives convergence to zero from every profile bounded by $\boldsymbol{\alpha}$, and the posterior-erasure formulas~\eqref{eq:mnposterior} then give the same conclusion for both posterior erasures.
\end{proof}

\subsection{Potential and classification of fixed points}\label{sec:mnpotential}
The one-dimensional capacity results for $(\ell,2,2)$ and $(\ell,3,3)$ MN ensembles use fixed-point potentials~\cite{MN2013,MN2014}. We retain that classification and prove the sign needed by Theorem~\ref{thm:transfer} (transfer to local shortening) for general $r,g$. The potential is the specialization of the vector construction~\cite{Vector2012} used in~\cite[Section II-C]{MN2014}.

\begin{definition}[MN potential and fixed-point classes]\label{def:mnpotential}
For the check update~\eqref{eq:mncheck} and variable update~\eqref{eq:mnf}, take
\begin{align}
 D&=\operatorname{diag}(r,g),\qquad
 F_{\mathrm{MN}}(\mathbf{y};\eps)=\frac r\ell y_1^\ell+\eps y_2^g,\notag\\
 G_{\mathrm{MN}}(a,b)&=ra+gb+(1-a)^r(1-b)^g-1,\label{eq:mnpotingredients}\\
 U_{\mathrm{MN}}(\mathbf{x};\eps)&=Q(\mathbf{x})^{\mathsf T}D\mathbf{x}-G_{\mathrm{MN}}(\mathbf{x})
                    -F_{\mathrm{MN}}(Q(\mathbf{x});\eps).
 \label{eq:mnpotential}
\end{align}
They satisfy $\nabla F_{\mathrm{MN}}=Df$ and $\nabla G_{\mathrm{MN}}=DQ$, as required in~\eqref{eq:primitives}. Define
\begin{align}
 \mathcal F_{\mathrm{MN}}(\eps)
 &=\{\mathbf{x}\in[0,\mathbf{b}_{\max}]\setminus\{0\}:\mathbf{x}=T(\mathbf{x})\},\notag\\
 \mathcal F_{\mathrm{nt}}^{\mathrm{MN}}(\eps)
 &=\{(a,b)\in\mathcal F_{\mathrm{MN}}(\eps):0<a<1\}.
 \label{eq:mnfixedset}
\end{align}
Following the BEC classification in~\cite[Section III]{MN2014}, $(0,0)$ and $(1,\eps)$ are the trivial fixed points; the elements of $\mathcal F_{\mathrm{nt}}^{\mathrm{MN}}(\eps)$ are the nontrivial fixed points.
\end{definition}

For $0<\eps<1$, the first fixed-point equation with $a=0$ forces $b=0$, while $a=1$ gives $b=\eps$. Thus these classes exhaust the fixed points. At a nontrivial fixed point, $0<b<\eps<1$. At the nonzero trivial fixed point, $Q(1,\eps)=(1,1)^{\mathsf T}$, and the individual terms in~\eqref{eq:mnpotential} give
\begin{equation}\label{eq:mnbadpotential}
 U_{\mathrm{MN}}(1,\eps;\eps)
 =(r+g\eps)-(r+g\eps-1)-(r/\ell+\eps)
 =1-\frac r\ell-\eps.
\end{equation}
The capacity gap is therefore the potential of this endpoint. It remains to prove strict positivity on the nontrivial fixed-point set $\mathcal F_{\mathrm{nt}}^{\mathrm{MN}}(\eps)$ defined in~\eqref{eq:mnfixedset}.

\subsection{An elementary positivity proof for all degrees}\label{sec:mnpositivity}
The first fixed-point equation eliminates one message variable, as in~\cite[Section III]{MN2014}. The resulting expression permits two comparisons of degrees. These compare values of an algebraic expression along fixed-point branches; the induced channel parameter may change during a comparison.

\begin{lemma}[Nontrivial fixed-point potentials]\label{lem:mnpositive}
Let $\ell>r\geq2$, $g\geq2$, and $0<\eps<1$. Then $U_{\mathrm{MN}}(\mathbf{x};\eps)>0$ for every $\mathbf{x}\in\mathcal F_{\mathrm{nt}}^{\mathrm{MN}}(\eps)$.
\end{lemma}

The nontrivial fixed-point set $\mathcal F_{\mathrm{nt}}^{\mathrm{MN}}(\eps)$ in~\eqref{eq:mnfixedset} restricts the message coordinates $a,b$ to actual homogeneous fixed points. The branch variable $t=a^{1/(\ell-1)}$ and the factors $A,P,\chi_g,C,z_g$ in~\eqref{eq:mnbranch} rewrite their potential as the scalar expression $V_{r,g}(t)$ in~\eqref{eq:mnbranchpotential}. The degree comparison~\eqref{eq:mngeneralpositive} concerns this expression with $\ell$ held fixed. Intermediate values need not describe fixed points at the original channel parameter. Positivity is proved on the algebraic domain used in the comparison and then applied to the actual fixed point.
\begin{proof}
Set $n=\ell-1\geq2$, the punctured-variable excess degree from the MN update~\eqref{eq:mnf}. For a nontrivial fixed point $(a,b)\in\mathcal F_{\mathrm{nt}}^{\mathrm{MN}}(\eps)$, with punctured and transmitted erasure coordinates as in~\eqref{eq:mncheck}, define
\begin{align}
 t&=a^{1/n},\qquad A=1-a=1-t^n,\qquad
 P=\frac{1-t}{A^{r-1}},\notag\\
 \chi_g&=P^{1/g}=1-b,\qquad C=A(1-t),\qquad z_g=1-C/\chi_g.
 \label{eq:mnbranch}
\end{align}
Here $0<t<1$, $0<P<1$, $Q_1=t$, $Q_2=z_g$, and $\eps z_g^{g-1}=b$. Substitution in~\eqref{eq:mnpotential}, using the last identity to eliminate $\eps$, gives
\begin{align}
 U_{\mathrm{MN}}(a,b;\eps)
 &=1-ra+\frac{rn}{n+1}at-C
       -b\bigl[1+(g-1)C/\chi_g\bigr]\notag\\
 &=\chi_g-(g-1)\frac C{\chi_g}+(g-2)C-ra+\frac{rn}{n+1}at
 =:V_{r,g}(t).
 \label{eq:mnbranchpotential}
\end{align}
We prove positivity first at $r=g=2$, then compare the general expression with that case.

At $r=g=2$, let $\sigma_n(t)=\sum_{j=0}^{n-1}t^j$. Since $P=1/\sigma_n(t)$, the branch-potential expression~\eqref{eq:mnbranchpotential} reduces to
\begin{equation}\label{eq:mn22potential}
 V_{2,2}(t)=t^n\left[\frac{2-t^n}{\sqrt{\sigma_n(t)}}-2+\frac{2n}{n+1}t\right].
\end{equation}
The two sides of the desired inequality
\begin{equation}\label{eq:mn22target}
 (n+1)(2-t^n)>2(n+1-nt)\sqrt{\sigma_n(t)}
\end{equation}
are positive, so squaring preserves its direction. Their squared difference is
\begin{align}
 P_n(t)&=(n+1)^2(2-t^n)^2-4(n+1-nt)^2\sigma_n(t)\notag\\
 &=4(n^2-1)t-4\sum_{j=2}^{n}t^j-4n^2t^{n+1}+(n+1)^2t^{2n}.
 \label{eq:mn22polynomial}
\end{align}
Normalize the squared difference $P_n(t)$ from~\eqref{eq:mn22polynomial} by setting $J_n(t)=P_n(t)/t$; this preserves its sign because $t>0$. Since $t^{2n-2}\leq t^{n-1}$ for $0<t\leq1$,
\begin{align}
 J_n'(t)
 &=-4\sum_{j=1}^{n-1}jt^{j-1}-4n^3t^{n-1}
       +(n+1)^2(2n-1)t^{2n-2}\notag\\
 &\leq-4\sum_{j=1}^{n-1}jt^{j-1}
       -(n-1)^2(2n+1)t^{n-1}<0.
 \label{eq:mn22derivative}
\end{align}
The coefficient identity used here is $4n^3-(n+1)^2(2n-1)=(n-1)^2(2n+1)$. Therefore
\begin{equation}\label{eq:mn22positive}
 J_n(t)>J_n(1)=(n-1)^2>0\qquad(0<t<1).
\end{equation}
The base-case positivity inequality~\eqref{eq:mn22target} and potential formula~\eqref{eq:mn22potential} yield $V_{2,2}(t)>0$ for every $0<t<1$. This inequality does not require the channel parameter induced by that branch point to lie in $[0,1]$.

Next keep $g=2$ and increase $r$. With $z=-\tfrac12\log A>0$, a finite geometric sum gives
\begin{align}
 \chi_2-C/\chi_2
 &=a\sqrt{1-t}\sum_{k=0}^{r-1}A^{k-(r-1)/2}
 =a\sqrt{1-t}\frac{\sinh(rz)}{\sinh z},\notag\\
 \frac{V_{r,2}(t)}{ra}
 &=\sqrt{1-t}\frac{\sinh(rz)}{r\sinh z}-1+\frac n{n+1}t.
 \label{eq:mnrcomparison}
\end{align}
For positive real $r$, the derivative of $\sinh(rz)/r$ has numerator $rz\cosh(rz)-\sinh(rz)$. This numerator is positive: as a function of $x=rz>0$, it vanishes at zero and has derivative $x\sinh x>0$. Hence
\begin{equation}\label{eq:mnrmonotone}
 V_{r,2}(t)\geq\frac r2V_{2,2}(t)>0\qquad(r\geq2).
\end{equation}

Finally fix $r,t$ from the actual fixed point and increase $g$. The condition $0<P<1$ in~\eqref{eq:mnbranch} remains valid throughout this comparison. The formula defines a smooth extension in real $g\geq2$, since $P>0$; only its integer values are needed for code degrees. In this extension, set $h=-\log P>0$ and $x=h/g$, and differentiate~\eqref{eq:mnbranchpotential}:
\begin{equation}\label{eq:mngderivative}
 \frac{\partial V_{r,g}(t)}{\partial g}
 =\frac{x}{g}e^{-x}
 +C\left[1-e^x+\left(1-\frac1g\right)xe^x\right].
\end{equation}
If the bracket is nonnegative, the derivative is positive. If it is negative, use $C=A^rP\leq P=e^{-gx}$. Replacing $C$ by this larger value multiplies a negative quantity and therefore gives a lower bound. It follows that
\begin{align}
 e^{(g-1)x}\frac{\partial V_{r,g}(t)}{\partial g}
 &\geq\frac{x}{g}e^{(g-2)x}+e^{-x}-1+\left(1-\frac1g\right)x\notag\\
 &\geq e^{-x}-1+x>0.
 \label{eq:mngmonotone}
\end{align}
The second inequality uses $g\geq2$, and the strict inequality uses $e^{-x}>1-x$ for $x>0$. Thus the comparison chain is
\begin{equation}\label{eq:mngeneralpositive}
 V_{r,g}(t)\geq V_{r,2}(t)\geq\frac r2V_{2,2}(t)>0.
\end{equation}
Together with the fixed-point potential representation~\eqref{eq:mnbranchpotential}, this proves the claim for every nontrivial fixed point.
\end{proof}

\subsection{Completion of the MN decoding theorem}\label{sec:mncompletion}
The fixed-point calculation now supplies Step 1, so the common flat-boundary, ball, and torus arguments give Steps 2--4 of Figure~\ref{fig:proofstory}. The rate calculation in Section~\ref{sec:mnrate} completes Step 5.

Lemma~\ref{lem:mnpositive} (nontrivial fixed-point positivity) handles $\mathcal F_{\mathrm{nt}}^{\mathrm{MN}}(\eps)$, and~\eqref{eq:mnbadpotential} handles the remaining nonzero fixed point. Consequently
\begin{equation}\label{eq:mnpositivecriterion}
 U_{\mathrm{MN}}(\mathbf{x}_\star;\eps)>0\quad
 \bigl(\mathbf{x}_\star\in\mathcal F_{\mathrm{MN}}(\eps),\quad0<\eps<1-r/\ell\bigr).
\end{equation}
The remaining conditions in Definition~\ref{def:admissible} (potential and contraction assumptions) follow from the check-update formulas~\eqref{eq:mncheck} and potential primitives~\eqref{eq:mnpotingredients}. Lemma~\ref{lem:mnbox} (contracting rectangle) supplies the small-message contraction condition~\eqref{eq:smallbox}, including $g=2$. Theorem~\ref{thm:transfer} (transfer to local shortening) therefore gives finite coupling and shortening sizes in every fixed dimension and proves the density-evolution assertion of Theorem~\ref{thm:mnmain} (MN decoding and capacity) for positive $\eps$. At $\eps=0$, compare with any positive parameter below $1-r/\ell$ by channel monotonicity. The message convergence gives posterior convergence through the punctured- and transmitted-bit posterior formulas~\eqref{eq:mnposterior}.

The dependence on dimension is confined to the directional and geometric arguments of Sections~\ref{sec:core} and~\ref{sec:geometry}. The fixed-point calculation~\eqref{eq:mngeneralpositive} uses only the homogeneous BEC recursion. Thus its general-degree conclusion does not require a separate wave-existence theorem for each pair of message types.

\subsection{Actual transmitted rate and capacity}\label{sec:mnrate}
Step 5 in Figure~\ref{fig:proofstory} requires a rate statement in addition to decoding. The shortened fraction vanishes as the torus grows, but puncturing also requires control of a projection rank. The decoding result provides that control through recovery of punctured variables.

After puncturing, the rate must be computed for the projected transmitted code. Let $V=L^d$, $V_S=s(w)^d$, and write the shortened parity-check matrix as $[H_P\ H_T]$. The column counts and the number of checks are
\begin{equation}\label{eq:mncounts}
 n_P=\frac r\ell M(V-V_S),\qquad n_T=M(V-V_S),\qquad m_c=MV.
\end{equation}
All ranks are over the binary field. The transmitted code is the projection of $\ker[H_P\ H_T]$, so its dimension satisfies
\begin{align}
 k_T&=n_T-\operatorname{rank}[H_P\ H_T]+\operatorname{rank}H_P,\label{eq:mnprojectedrank}\\
 \frac{k_T}{n_T}&\geq\frac r\ell-\frac{V_S}{V-V_S}
                  -\frac{\operatorname{nullity}H_P}{n_T}.
 \label{eq:mnactualrate}
\end{align}

In the count formulas~\eqref{eq:mncounts}, $n_P$ and $n_T$ count the remaining punctured and transmitted variables; $m_c$ counts check rows. The matrices $H_P$ and $H_T$ contain their respective columns. Puncturing can identify different full codewords that have the same transmitted coordinates. The quantity $\operatorname{nullity}H_P$, the dimension of the solution space of $H_P\mathbf p=0$, counts this lost freedom. Thus the rate lower bound~\eqref{eq:mnactualrate} separates the cost of shortening, $V_S/(V-V_S)$, from the cost of this ambiguity, $\operatorname{nullity}H_P/n_T$. The first vanishes by enlarging the torus; Lemma~\ref{lem:peelingrank} (recovery and punctured nullity) controls the second using recovery of punctured variables.
Indeed, the kernel of the projection consists of $(\mathbf{p},0)$ with $H_P\mathbf{p}=0$. Subtracting its dimension gives~\eqref{eq:mnprojectedrank}, and the bound $\operatorname{rank}[H_P\ H_T]\leq m_c$ gives~\eqref{eq:mnactualrate}.

\begin{lemma}[Finite-iteration approximation for the socket ensemble]\label{lem:finiteDE}
Fix the dimension, degrees, $L,w$, shortening set, and a finite number $t$ of SPA iterations in Definition~\ref{def:ensemble} (socket allocation and shortening). For each unshortened variable type $a$, let $E_{a,t}(\Gamma,\eps)$ be its average posterior erasure probability on a graph $\Gamma$, with the expectation over the BEC already taken. Let $\bar e_{a,t}(\eps)$ be the corresponding density-evolution average over unshortened positions. There is a finite constant $C_t$, independent of $M$ and $\eps\in[0,1]$, such that, for all sufficiently large admissible $M$,
\begin{equation}\label{eq:finiteDE}
 \left|\mathbb E_\Gamma E_{a,t}(\Gamma,\eps)-\bar e_{a,t}(\eps)\right|\leq\frac{C_t}{M}.
\end{equation}
For MN codes, this includes the punctured type and the experiment in which all transmitted variables are revealed, corresponding to $\eps=0$.
\end{lemma}
\begin{proof}
We specialize the computation-tree argument for sparse random ensembles~\cite{RU2001} to the balanced offset groups of Definition~\ref{def:ensemble} (socket allocation and shortening). Explore the neighborhood needed for $t$ iterations from a uniformly chosen variable of the specified type. Bounded degrees and fixed $t$ bound the number of exposed sockets by a finite $D_t$, independent of $M$. Each position, socket type, and offset group has size at least $c_0M$ for a fixed $c_0>0$. At any exposure, the chance of encountering an already exposed variable or check is therefore $O(D_t/M)$. Summing over exposures bounds the probability of any collision by $O(D_t^2/M)$.

Balanced offset labels are sampled without replacement. Conditional on at most $D_t$ previous exposures, their probabilities differ from independent uniform offsets by $O(D_t/M)$. They can consequently be coupled to independent offsets with total failure probability $O(D_t^2/M)$. Outside these two failure events, the explored neighborhood agrees with the multi-type computation tree. Independent BEC observations on distinct transmitted variables give exactly the two spatial averages in Definition~\ref{def:DE} (shortened density evolution); shortened variables supply known values and punctured variables supply erasures. The posterior erasure indicator is in $[0,1]$, so the coupling failure probability bounds the error in its expectation. Averaging over roots proves~\eqref{eq:finiteDE}. Neither failure event depends on the BEC erasure probability. The argument keeps $L,w,t$ fixed while $M$ grows.
\end{proof}

\begin{lemma}[Recovery bounds the punctured nullity]\label{lem:peelingrank}
For a fixed MN graph $\Gamma$, reveal all transmitted variables. Let $U_{P,t}(\Gamma)$ be the number of punctured variables unrecovered after $t$ SPA iterations. This count depends only on the graph and iteration count, and satisfies
\begin{equation}\label{eq:peelingrank}
 \operatorname{nullity}H_P\leq U_{P,t}(\Gamma).
\end{equation}
\end{lemma}
\begin{proof}
On the BEC, SPA recovery can be ordered as successive removals of variables determined by a check with a single remaining unknown, the usual peeling decoder~\cite[Section II]{Luby2001}. Order the recovered punctured variables as $v_1,\ldots,v_k$, respecting this recovery order, and choose a witnessing check for each. Such a check has coefficient one on $v_i$ and zero on every punctured variable not yet recovered. The witnessing checks are distinct: after recovering their last unknown they cannot recover another variable. Their rows, restricted to columns $v_1,\ldots,v_k$, form a triangular matrix with diagonal entries one. Hence $\operatorname{rank}H_P\geq k=n_P-U_{P,t}$. This proves~\eqref{eq:peelingrank}. The argument also permits parallel edges: a witnessing SPA check has only one unknown incident socket, so its coefficient on that variable is one after reducing multiplicities modulo two.
\end{proof}

\begin{proposition}[One graph with both rate and reliability]\label{prop:finiteMN}
For every $0\leq\eps<1-r/\ell$ and $\zeta,\delta>0$, there is a finite graph $\Gamma$ and a finite iteration count $t$ such that
\begin{equation}\label{eq:mnratedelivery}
 \frac{k_T}{n_T}\geq\frac r\ell-\zeta,
 \qquad E_{T,t}(\Gamma,\eps)\leq\delta.
\end{equation}
Here the second quantity averages over the BEC for this same selected graph.
\end{proposition}

Two separate good ensemble averages could in principle be attained by different graphs. The quantities $X_t,Y_t$ in~\eqref{eq:selectionvariables} track the normalized unrecovered punctured count and the transmitted-bit erasure probability on the same graph $\Gamma$. The combined estimate~\eqref{eq:selectionmarkov} ensures that both targets hold for at least one graph. This simultaneous selection connects the nullity bound~\eqref{eq:peelingrank} to reliable transmission, rather than producing unrelated rate and decoding examples.
\begin{proof}
Choose $w$ and $s(w)$ by the density-evolution part of Theorem~\ref{thm:mnmain} (MN decoding and capacity), already proved in Section~\ref{sec:mncompletion}, and then fix $L$ large enough that $L>s(w)+2w$ and $V_S/(V-V_S)\leq\zeta/2$. For the random graph define two nonnegative quantities
\begin{equation}\label{eq:selectionvariables}
 X_t(\Gamma)=\frac{U_{P,t}(\Gamma)}{n_T},\qquad
 Y_t(\Gamma)=E_{T,t}(\Gamma,\eps).
\end{equation}
Revealing all transmitted variables can only help recovery. Thus the punctured density-evolution erasure at parameter zero and the transmitted erasure at parameter $\eps$ both tend to zero. First take $t$ large enough and then $M$ large enough in Lemma~\ref{lem:finiteDE} (finite-iteration approximation) to obtain
\begin{equation}\label{eq:selectionexpectations}
 \mathbb E_\Gamma X_t\leq\frac{\zeta}{8},\qquad
 \mathbb E_\Gamma Y_t\leq\frac{\delta}{4}.
\end{equation}
For the first estimate use $n_P/n_T=r/\ell$ to convert the punctured-type average to the normalization in $X_t$. Markov's inequality now gives
\begin{equation}\label{eq:selectionmarkov}
 \Pr_\Gamma\!\left\{\frac{2X_t}{\zeta}+\frac{Y_t}{\delta}\geq1\right\}
 \leq\mathbb E_\Gamma\!\left[\frac{2X_t}{\zeta}+\frac{Y_t}{\delta}\right]\leq\frac12.
\end{equation}
Consequently some graph has $X_t<\zeta/2$ and $Y_t<\delta$ simultaneously. Lemma~\ref{lem:peelingrank} (recovery and punctured nullity) and the transmitted-rate lower bound~\eqref{eq:mnactualrate} give its rate guarantee. This selection needs only expectations, not a concentration theorem.
\end{proof}

An entropy bound completes the actual-rate limit. Choose $\eps_j\uparrow1-r/\ell$, $\zeta_j,\delta_j\downarrow0$, and a code and iteration count from Proposition~\ref{prop:finiteMN} (simultaneous rate and reliability) for each $j$. The selection proof allows $L$ and admissible $M$ to be increased before selecting the graph. For the $j$th code, require also $L_j\geq j$ and $M_j\geq j$, so that the transmitted block length $n_{T,j}=M_j(L_j^d-s(w_j)^d)$ tends to infinity. For each $j$, suppress these subscripts and let $\mathbf X$ be uniform on the projected transmitted code and $\mathbf Y$ its BEC output. The entropy chain rule and memorylessness give the usual channel-capacity bound~\cite{Shannon1948}:
\begin{equation}\label{eq:mninformationbound}
 I(\mathbf X;\mathbf Y)\leq\sum_{i=1}^{n_T}I(X_i;Y_i)\leq n_T(1-\eps_j).
\end{equation}
For each output $\mathbf y$, SPA correctly determines all but $N_{T,t_j}^{\mathrm{unrec}}(\mathbf y)$ transmitted coordinates. The conditional support therefore has size at most $2^{N_{T,t_j}^{\mathrm{unrec}}(\mathbf y)}$. Averaging its logarithm bounds the conditional entropy by $n_T\delta_j$. Since $H(\mathbf X)=k_T$, we obtain
\begin{equation}\label{eq:mncapacitysandwich}
 \frac r\ell-\zeta_j\leq\frac{k_T}{n_T}
 \leq1-\eps_j+\delta_j\longrightarrow\frac r\ell.
\end{equation}

The tolerance $\delta_j$ in the simultaneous guarantee~\eqref{eq:mnratedelivery} bounds the expected fraction of unrecovered transmitted bits. Therefore the remaining uncertainty obeys $H(\mathbf X\mid\mathbf Y)\leq n_T\delta_j$. Combining this with the mutual-information bound~\eqref{eq:mninformationbound} in $H(\mathbf X)=I(\mathbf X;\mathbf Y)+H(\mathbf X\mid\mathbf Y)$ supplies the upper rate bound. The lower bound comes from the same graph selected by Proposition~\ref{prop:finiteMN} (simultaneous rate and reliability).

The erasure pattern of BEC SPA is independent of the transmitted codeword, so the probabilities in Proposition~\ref{prop:finiteMN} (simultaneous rate and reliability) apply to this uniform input. For every fixed $\eps<1-r/\ell$, eventually $\eps\leq\eps_j$; coupling the erasures makes the same code sequence reliable at $\eps$. This proves Theorem~\ref{thm:mnmain} (MN decoding and capacity) in the average bit-erasure sense. The construction chooses geometry, then iterations, then section size; it makes no interchange of these limits.

\section{Discussion}\label{sec:discussion}
Positive potentials at nonzero homogeneous fixed points lead to flat-boundary advancement, ball expansion, and decoding throughout a torus from a finite shortened region. The first implication uses the endpoint identity and removal of an auxiliary cap; the geometric implications use uniform directional estimates and finite dependence. Small-message contraction completes decoding. Regular LDPC codes satisfy these conditions below their potential threshold. For MN codes, the degree comparisons establish them for all $\ell>r\geq2$, $g\geq2$ below the BEC capacity boundary, and the transmitted-rate analysis completes the capacity statement.

The sufficient sizes in Section~\ref{sec:quantitative} still depend on finite advancement times whose existence is proved analytically. Useful numerical bounds on these times, smaller shortened regions, and finite-length block-erasure estimates are further questions. The constants may grow with the spatial dimension; the theorem fixes $d$ before taking larger tori. Extending the local-region argument to channels with non-scalar message densities would require a corresponding endpoint identity and a way to control finite-time directional perturbations in that density space.

\end{document}